\documentclass[draftclsnofoot,twoside,onecolumn,letter,12pt]{IEEEtran}
\usepackage{microtype}
\usepackage{capt-of}
\usepackage{siunitx}
\usepackage{graphicx}
\usepackage{amssymb}
\usepackage{amsmath}
\usepackage{algorithm}
\usepackage{algpseudocode}
\usepackage{cite}
\usepackage{epstopdf}
\usepackage{mdwlist}
\usepackage{threeparttable}
\usepackage[all,2cell]{xy} \UseAllTwocells
\usepackage{booktabs}
\usepackage{fancyvrb}
\usepackage{balance}
\usepackage{tensor}
\usepackage{leftidx}
\usepackage{mathabx}
\usepackage{color}
\usepackage{csquotes}
\usepackage{pdflscape}
\usepackage{afterpage}
\usepackage{array}
\usepackage{cases}
\usepackage{bm}

\newtheorem{theorem}{Theorem}
\newtheorem{lemma}{Lemma}
\newtheorem{corollary}{Corollary}
\newtheorem{proposition}{Proposition}

\newtheorem{remark}{Remark}

\AtBeginDocument{}

\newcommand{\e}{\mathrm{e}}

\newcommand{\Tb}{T_\mathrm{b}}
\newcommand{\Pe}{P_\mathrm{b}}
\newcommand{\erfc}{\mathrm{erfc}}

\newcommand{\erf}{\mathrm{erf}}
\newcommand{\diff}{\mathrm{d}}
\newcommand{\Pt}{P_\theta}

\newcommand{\atx}{a_\mathrm{tx}}
\newcommand{\arx}{a_\mathrm{rx}}

\newcommand{\Tx}{\mathrm{Tx}}
\newcommand{\Rx}{\mathrm{Rx}}

\newcommand{\TRx}{T_\mathrm{Rx }}

\begin{document}
	
	\title{\huge Diffusive Mobile MC with Absorbing Receivers: Stochastic Analysis and Applications}
	
	      \author{\IEEEauthorblockN{Trang Ngoc Cao, Arman Ahmadzadeh,  Vahid Jamali, Wayan Wicke,\\ Phee Lep Yeoh, Jamie Evans, and Robert Schober
	      		}
	      		\thanks{%
	      			This paper has been presented in part at IEEE ICC 2019.
	      		}%
%	\author{\IEEEauthorblockN{Trang Ngoc Cao\IEEEauthorrefmark{1}, Arman Ahmadzadeh\IEEEauthorrefmark{3}, Vahid Jamali\IEEEauthorrefmark{3}, Wayan Wicke\IEEEauthorrefmark{3}, \\ Phee Lep Yeoh\IEEEauthorrefmark{2}, Jamie Evans\IEEEauthorrefmark{1}, and Robert Schober\IEEEauthorrefmark{3}}
%		\IEEEauthorblockA{\IEEEauthorrefmark{1}Department of Electrical and Electronic Engineering, The University of Melbourne, Australia}
%		\IEEEauthorblockA{\IEEEauthorrefmark{2}School of Electrical and Information Engineering, The University of Sydney, Australia}
%		\IEEEauthorblockA{\IEEEauthorrefmark{3}Institute for Digital Communications, Friedrich-Alexander-University Erlangen-Nurnberg (FAU), Germany}
	}

	\maketitle
	\begin{abstract}

			This paper presents a  stochastic analysis of the time-variant channel impulse response (CIR) of a three dimensional diffusive mobile molecular communication (MC) system where  the transmitter, the absorbing receiver, and the molecules can freely diffuse. In our analysis, we derive the mean,  variance,   probability density function (PDF), and cumulative distribution function (CDF) of the CIR. We also derive the PDF and CDF of the probability  $p$ that a released molecule is absorbed  at the receiver during  a given time period. The obtained analytical results are employed for the design of drug delivery and  MC systems with imperfect channel state information. 
			For the first application, we exploit the mean and  variance of the CIR to   optimize  a controlled-release drug delivery system employing a mobile drug carrier. We   evaluate the performance of the proposed release  design based on the PDF and CDF of the CIR. We demonstrate significant savings in the amount of released drugs compared to a constant-release  scheme and reveal the necessity of accounting for the  drug-carrier's mobility to ensure reliable drug delivery. 
			For the second application, we exploit the PDF of the distance between the mobile transceivers and the CDF of $p$  to optimize three design parameters  of an MC system employing on-off keying modulation and threshold detection. Specifically, we optimize the detection threshold at the receiver, the release profile at the transmitter, and the time duration of a bit frame. We show that  the proposed optimal designs can significantly improve the system performance in terms of the bit error rate and the efficiency of molecule usage.

	\end{abstract}
	\section{Introduction}

	  As appropriate channel models are essential for the analysis and  design of molecular communication (MC) systems, MC channel modeling has been extensively studied in the literature, see \cite{JAW:18:Arxiv} and references therein. For example, the simple diffusive channel model of an unbounded three-dimensional (3D) MC system with  impulsive point release of  information carrying molecules \cite{YHT:14:CL} has been  widely used for system analysis and design, see \cite{FYE:16:CSTO,FRV:15:TNB}, and references therein. Diffusion channel models with drift \cite{KAE:12:N} and chemical reactions \cite{NCS:14:INB} have also been considered. However, most of the previously studied MC channel models assume static communication systems where the transceivers do not move. 
	  
	  Recently, many applications have emerged where the transceivers are mobile, including drug delivery \cite{CMM:17:ST},  mobile ad hoc networks \cite{GAA:12:MC}, and detection of mobile targets \cite{NOK:17:TC}. 
	  Hence, the modeling and design of mobile MC systems have gained considerable attention, e.g., see \cite{GAA:12:MC,NOK:17:TC,CLY:18:NB,LWLY:18:NB,AJS:18:COMT,JAW:18:Arxiv,HAA:17:CL,VHG:18:Arxiv,VJV:18:CISS}, and references therein. In \cite{GAA:12:MC}, a mobile ad hoc nanonetwork was considered where  mobile nanomachines collect  environmental information and deliver it to a mobile central control unit. The mobility of the nanomachines was described by a 3D model but information was only exchanged when  two nanomachines collided. In \cite{NOK:17:TC}, a leader-follower-based model for two-dimensional mobile MC networks for target detection with non-diffusive information molecules was proposed. The authors in  \cite{CLY:18:NB} considered  adaptive detection and inter-symbol interference (ISI) mitigation in mobile MC systems, while \cite{LWLY:18:NB} analyzed the mutual information and maximum achievable rate in such systems. However, the authors of \cite{CLY:18:NB} and \cite{LWLY:18:NB} did not provide a stochastic analysis of the time-variant channel but analyzed the system numerically. In \cite{AJS:18:COMT}, a comprehensive framework for modeling the time-variant channels of  diffusive mobile MC systems with diffusive transceivers was developed. However, all of the works  mentioned above assumed a passive receiver.
	  
	  On the other hand, for many MC applications, a fully absorbing receiver is considered to be a more realistic model compared to a passive receiver as it  captures the  interaction between the receiver and the information molecules, e.g., the conversion of the information molecules to a new type of molecule or the absorption and removal of the information molecules from the environment \cite{FYE:16:CSTO,YHT:14:CL}. Since the molecules are removed from the environment after being absorbed by the receiver, the channel impulse response (CIR) for absorbing receivers is a more complicated function of the distance between the transceivers and the receiver's radius compared to   passive receivers. Therefore, the  stochastic  analysis of  mobile MC systems with absorbing receiver is very challenging. For the fully absorbing receiver in diffusive mobile MC systems, theoretical expressions for the average distribution of the first hitting time, i.e., the mean of the CIR,  were derived for a \emph{one-dimensional (1D)} environment without drift in \cite{HAA:17:CL} and with drift in \cite{VHG:18:Arxiv}. Based on the 1D model in \cite{HAA:17:CL},   the error rate and channel capacity of the system were examined in \cite{VJV:18:CISS}. 
	  However, none of these works provides a statistical analysis of the time-variant  CIR of a 3D diffusive mobile  MC system with absorbing receiver. In this paper, we address this issue and exploit the obtained analytical results for the stochastic parameters  of  the time-variant MC channel for the design of  drug delivery  and MC systems.
	  
	  In  drug delivery systems,  drug molecules are carried to  diseased cell sites by  nanoparticle drug carriers, so that the drug is  delivered to the targeted site without  affecting healthy  cells \cite{CMM:17:ST}. After being injected or extravasated from the cardiovascular system into the tissue surrounding a targeted diseased cell site, the drug carriers may not be anchored at the targeted site but may move, mostly via diffusion \cite{SSR:14:BE,PNJ:99:BPJ, LEELYP:15:CES,WCX:14:PCL}. The diffusion of the drug carriers results in a time-variant absorption rate of the drugs  even if the  drug release rate is constant. 	Furthermore, experimental and theoretical studies have indicated that  the total drug dosage as well as the rate and time period of drug absorption by the receptors of the diseased cells  are critical factors in the healing process \cite{LEELYP:15:CES,KS:14:DD}. 
	  Therefore, to satisfy reduce drug cost, over-dosing, and negative side effects to healthy cells yet satisfy the treatment requirements, it is important to optimize the release profile of drug delivery systems  such that  the total amount of released drugs is minimized while  a desired rate of drug absorption  at the diseased site during a prescribed time period is achieved. To this end, the mobility of the drug carriers and the absorption rate of the drugs have to be accurately taken into account. This can be accomplished by exploiting the MC paradigm where the  drug carriers, diseased cells, and   drug molecules  are modeled as mobile transmitters,   absorbing receivers,  and signaling molecules, respectively~\cite{FYE:16:CSTO}. Release profile designs for drug delivery systems based on an MC framework  were proposed  in \cite{CPA:15:BME,FRV:15:TNB,SMA:17:NCN,SMH:18:NB}. However, in these works, the transceivers were fixed and only the movement of the drug molecules was considered. In this paper, we  exploit the  analytical results obtained for the stochastic parameters  of  the time-variant MC channel with absorbing receiver for the optimization of the  release profile of drug delivery systems with mobile drug carriers.
	
	In diffusive mobile  MC systems,  knowledge of the  CIR is needed for reliable communication design. However,   the CIR  may not always be  available in a diffusive mobile MC system due to the random movements of the transceivers. In particular, the distance between the transceivers  at the time of release, on which the CIR depends, may only be known at the start of a  transmission frame. In other words, the movement of the transceivers  causes the CSI to become outdated, which makes communication system design challenging. In this paper, we consider  a mobile MC system employing on-off keying and threshold detection and  optimize three design parameters to improve the system performance under imperfect CSI. First, we optimize the detection threshold at the receiver for minimization of the maximum bit error rate (BER) in a frame when the number of  molecules available for transmission is uniformly allocated to  each bit of the frame. Second, we  optimize the release profile at the transmitter, i.e.,  the optimal number of  molecules available for the transmission of each bit,  for  minimization of the maximum  BER in a frame given a fixed number of  molecules available for transmission of the entire frame. Third, we  maximize the frame duration under the constraint that the probability that a released molecule is absorbed by the receiver does not fall below a prescribed value. Such a design ensures that molecules are used efficiently as a molecule release occurs only if the released molecule is observed at the receiver with  sufficiently high probability. For the proposed design tasks, the results of the stochastic analysis of the transceivers'  positions and of the  probability that a  molecule is absorbed during  a given time period are exploited.

%	In this paper, we analyze the  time-variant  channel with an absorbing receiver and consider two important applications among various applications of diffusive mobile MC systems. %\cite{AJS:18:COMT,JAW:18:Arxiv}.  
	In summary, the main contributions of this paper are as follows:
	\begin{itemize}
		\item We provide a  statistical analysis of the time-variant channel of a 3D diffusive mobile MC system employing an absorbing receiver. In particular, we derive the   mean, variance, PDF, and CDF of the corresponding CIR. Moreover, we derive the PDF and CDF of the  probability that a  molecule is absorbed during  a given time period. The stochastic channel analysis is exploited for the design  of drug delivery and MC systems.
		
		\item For drug delivery systems, the release profile is optimized for the minimization of the amount of released drugs while ensuring that the absorption rate at the diseased cells  does not fall below a prescribed threshold for a given period of time. We show that the proposed design requires a significantly lower  amount of released drugs  compared to a design with  constant-release rate.
		\item For MC systems employing on-off
		keying modulation and threshold detection based on imperfect CSI,  we  optimize three design parameters, namely the detection threshold at the receiver, the release profile at the transmitter, and the time duration of a bit frame. Simulation  results show significant performance gains for the proposed designs in terms of BER and the efficiency of molecule usage compared to baseline systems with uniform molecule release and without limitation on time duration of a bit frame, respectively.
		
		\item Our results reveal that the transceivers' mobility  has a significant impact on the system performance and should be taken into account for MC system design.
		
	\end{itemize}
	
	We note that  the derived analytical results   for the time-variant CIR  of mobile MC systems with absorbing receiver are  expected to be  useful not only for the design of the drug delivery and MC systems considered in this paper but also for the design of  detection schemes and the evaluation of the  performance (e.g., the capacity and throughput) of such systems.
	
    This paper expands its  conference
	version \cite{CAJ:19:ICC}. In particular, the  analysis of the  probability that a  molecule is absorbed during  a given time period, the MC system design for imperfect CSI, and the corresponding simulation results are not included in \cite{CAJ:19:ICC}. 
	
	The remainder of this paper is organized as follows. In Section \ref{sec:2}, we introduce the considered diffusive mobile MC system with absorbing receiver and the time-variant channel model. In Section \ref{sec:3}, we provide the proposed statistical analysis of the time-variant channel. In Sections \ref{sec:4} and \ref{sec:5}, we apply the derived results for optimization of drug delivery and  MC systems with imperfect CSI, respectively. Numerical results are presented in Section \ref{sec:6}, and  Section \ref{sec:7} concludes the paper.

	\section{General System and Channel Model}\label{sec:2}

	In this section, we first introduce the model for a general diffusive mobile MC system  with  absorbing receiver. Subsequently, we specialize the model to drug delivery and communication systems with imperfect CSI. Finally, we define the time-variant CIR and the received signal.
	
		\subsection{System Model} \label{sub2:1}
		
		\begin{figure}[t!] 		
			\centering
			\includegraphics[width=0.55\textwidth, trim=0 0 0 600,clip=true]{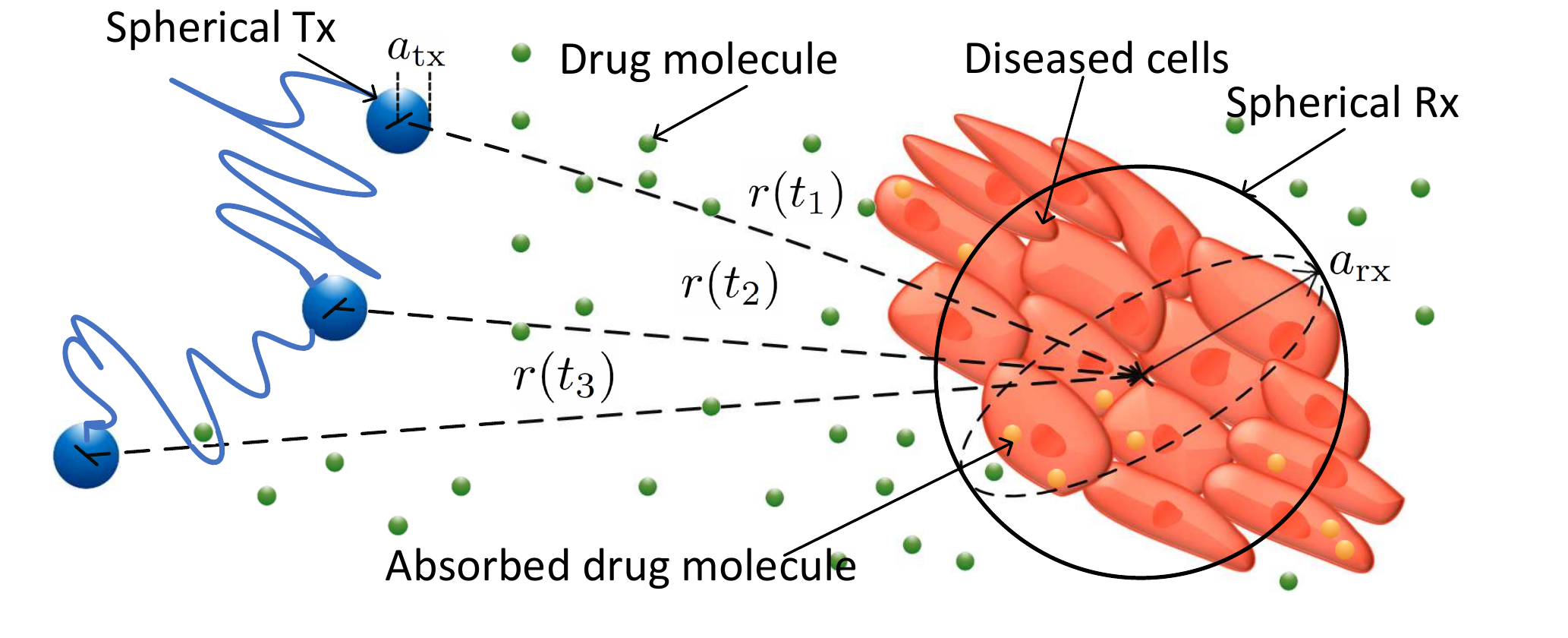}
			\caption{
				System model for  drug delivery.  The drug carrier and the diseased cells of a tumor are modeled as  diffusive spherical transmitter ($\Tx$) and  spherical absorbing receiver ($\Rx$), respectively.  The drug molecules are absorbed by  $\Rx$, when they hit its surface. Different distances between $\Tx$ and $\Rx$ over time are due to $\Tx$'s diffusion.
			}
			\label{fig:2}
		\end{figure}
		
%			\subsubsection{General System Model} \label{sub21:1}
			We consider a   linear diffusive mobile MC system in an unbounded 3D environment with constant temperature and viscosity. The system comprises  one mobile spherical  transparent transmitter, denoted by $\Tx$, with radius $\atx$, one mobile spherical  absorbing receiver, denoted by $\Rx$, with radius $\arx$, and the signaling molecules of type $\mathrm{X}$.

			The movements of  $\Tx$, $\Rx$, and  $\mathrm{X}$ molecules  are assumed to be mutually independent  and  follow  Brownian motion with diffusion coefficients $D_\Tx$, $D_\Rx$, and $D_\mathrm{X}$, respectively. This assumption, which was  also made in \cite{AJS:18:COMT} and \cite{HAA:17:CL}, is motivated by the fact that the mobility of small objects is governed by Brownian motion.
						
			We assume that  $\Tx$ releases molecules at its center  instantaneously and discretely during the considered period of time denoted by $T$. Let $t_i$ and $\Tb$ denote the time instant  of the $i$-th release and the duration of the interval between the $i$-th  and  the $(i+1)$-th release, respectively. We have $t_i=(i-1)\Tb$ and $i\in \{1,\dots,I\}$, where $I$ is the total number of releases during $T$. 			
			We denote the time-varying distance between the centers of  $\Tx$ and $\Rx$ at time $t$ by $r(t)$.  			
			Furthermore, let $\alpha_i$ and $A=\sum_{i=1}^I \alpha_i$  denote the number of  molecules released at time $t_i$ and the total number of molecules released during $T$, respectively. For concreteness, we specialize the considered general model to two application scenarios.
		
		\subsubsection{Drug Delivery Systems}
		
		A drug delivery system comprises  a drug carrier releasing drug molecules and diseased cells absorbing them. We model the drug carrier and diseased cells as   $\Tx$ and $\Rx$ of the general MC system, respectively, see Fig.~\ref{fig:2}. The drug carriers in  drug delivery systems are typically  nanoparticles, such as  spherical polymers or polymer chains, having a size not smaller than $\SI{100}{\nano\metre}$  \cite{PNJ:99:BPJ}. Moreover,  drug carriers  are designed to carry drug molecules and interaction with the drug or the receiver is not intended.  Hence, the drug carriers can be modeled as mobile spherical transparent transmitters, $\Tx$.
		When the drug molecules hit the tumor, they are absorbed by receptors on the surface of the diseased cells  \cite{LEELYP:15:CES,KS:14:DD}. For convenience, we model the tumor as a spherical absorbing receiver, $\Rx$. In reality, the colony of cancer cells may potentially  have a different geometry, of course. However, as an abstract approximation, we model the cancer cells as one effective spherical receiver with radius $\arx$ and with a surface area equivalent to the total surface area of the tumor (see Fig.~\ref{fig:2}). Hence, the absorption on  the actual and the modeled surfaces is expected to be comparable \cite{SSR:14:BE}.	
		
		 In a drug delivery system, the drug carriers can be directly injected or extravasated from the blood into the interstitial tissue near the diseased cells, where they start to move. We assume that the injection position can be estimated and thus $r(t=0)$ is known. The movement of the drug carrying nanoparticles in the tissue is caused by diffusion and convection mechanisms but diffusion is expected to be  dominant in  most cases \cite{SSR:14:BE,PNJ:99:BPJ, LEELYP:15:CES,WCX:14:PCL}. At the tumor site, the drug carrier releases drug molecules of type $\mathrm{X}$, which also diffuse in the tissue \cite{LEELYP:15:CES}. Hence,  we can adopt Brownian motion to model the diffusion of   $\Tx$ and    $\mathrm{X}$ molecules with diffusion coefficients $D_\Tx$, and $D_\mathrm{X}$, respectively \cite{CMM:17:ST}. We consider a rooted tumor and thus $D_\Rx=0$, which is a special case of the considered general system model.
		
	We assume the instantaneous and discrete release of drugs.  After releasing for a period, the drug carrier may be removed by  blood circulation or run out of drugs. Thus, for drug delivery systems, $T$, $\Tb$, $t_i$, and $I$ denote the release period of the drug, the duration  of the interval between two releases, the release instants of the drug molecules,   and the number of releases, respectively. A continuous release can  be approximated by letting $\Tb \rightarrow 0$, i.e., $I \rightarrow \infty$. Moreover, $A$ and $\alpha_i$ denote the total number of drug molecules released during $T$ and the number of drug molecules  released at time $t_i$, respectively.

	\subsubsection{Molecular Communication System}
	
	 For the considered MC system, we  assume  Brownian motion of the transceivers and signaling molecules. We assume multi-frame communication between mobile $\Tx$ and $\Rx$ with instantaneous molecule release for each bit transmission. Hence, $T$,  $\Tb$, $t_i$, and  $I$ denote the duration of a bit frame, the duration of one bit interval,  the beginning of the $i$-th bit  interval, and the number of bits in a frame, respectively. For an arbitrary bit frame, let  $b_i$, $i\in \{1,\dots,I\}$, denote the $i$-th bit in the bit frame. We assume that symbols $0$ and $1$ are	transmitted independently and with equal probability. Thus, the probability of transmitting $\tilde{b}_i$ is $\Pr(\tilde{b}_i)=1/2$, where $\Pr(\cdot)$ denotes  probability and $\tilde{b}_i\in \left\{0,1\right\}$ is a realization of $b_i$. 	 We assume that  on-off keying modulation is employed. At time $t_i$, $\Tx$ releases $\alpha_i$  molecules to transmit bit 1 and no molecules for bit 0.  Then, $A=\sum_{i=1}^I \alpha_i$  is  the total number of molecules available for transmission in a given bit frame.

	 \subsection{Time-variant CIR and Received Signal} \label{sub2:2}
	 
	 Considering again the general system model, we now model the channel  between $\Tx$ and $\Rx$ as well as the received signals at $\Rx$ for drug delivery and MC systems, respectively.

		\subsubsection{Time-variant CIR}
		Let $h(t,\tau)$ denote the hitting rate, i.e., the absorption rate of a given molecule, at time $\tau$ after its release at time $t$ at the center of $\Tx$. Then, for an infinitesimally small observation window $\Delta \tau$, i.e., $\Delta \tau \rightarrow 0$, we can interpret $h(t,\tau)\Delta \tau$ as the probability of absorption of a molecule  by  $\Rx$ between times $\tau$ and $\tau+\Delta \tau$ after its release at time $t$.  The hitting rate $h(t,\tau)$ is also referred to as the CIR since it completely characterizes the time-variant channel, which is assumed to be linear.

		For a given distance between $\Tx$ and $\Rx$, $r(t)$, the CIR  $h(t,\tau)$ of a diffusive mobile MC system at time $\tau$ is given by \cite{YHT:14:CL,JAW:18:Arxiv}
		\begin{align} \label{eq:6}
			h(t,\tau)=\frac{\arx}{\sqrt{4\pi D_1 \tau^3}}\left(1-\frac{\arx}{r(t)}\right)\exp\left(-\frac{\left(r(t)-\arx\right)^2}{4 D_1 \tau}\right), \hspace{1cm} \tau>0,
		\end{align}
		where $h(t,\tau)=0$, for $\tau\leq 0$. Here, $D_1$ is  the effective diffusion coefficient capturing the
		relative motion of the signaling molecules and  $\Rx$, i.e., $D_1=D_X+D_\Rx$,
		see \cite[Eq.~(8)]{AJN:17:CL}. In the considered MC system, due to the motion of the transceivers, the distance $r(t)$ is a random variable, and thus, the CIR  $h(t,\tau)$ is time-variant and should be modeled as a stochastic process \cite{AJS:18:COMT}.

		\subsubsection{Received Signal for Drug Delivery System}

		In drug delivery, the absorption rate ultimately determines the  therapeutic impact of   the drug \cite{LEELYP:15:CES,KS:14:DD}. Thus, we formally define the absorption rate as the desired received signal, and make achieving a desired absorption rate  the objective for  system design. 
		Recall  that   $h(t,\tau)\Delta \tau$, $\Delta \tau\rightarrow 0$,  is the probability of absorption of a molecule  by  $\Rx$ between times $\tau$ and $\tau+\Delta \tau$ after the release at time $t$. If  $\alpha_i$ molecules are released at  $\Tx$ at time $t_i$, the expected number of molecules absorbed at  $\Rx$  between times $t$ and $t+\Delta t$, for $\Delta t \rightarrow 0$, due to this release is  $\alpha_i h( t_i,t-t_i) \Delta t$.  During the period $[0,t]$,   the total number of released drug molecules is   $A_t=\sum_i \alpha_i, \forall i | t_i < t,$  and the expected number of drug molecules absorbed between times $t$ and $t+\Delta t$, for $\Delta t \rightarrow 0$ is given by  $s(t)=\sum_i\alpha_i h( t_i,t-t_i) \Delta t, \forall i | t_i < t$. Let $g(t)$ denote the absorption rate of drug molecules $\mathrm{X}$    at  $\Rx$  at time $t$, i.e., $g(t) =s(t)/\Delta t$, $\Delta t \rightarrow 0$.  Then, we have
		\begin{align}\label{eq:5a}
		g\left(t\right)=\underset{\forall i | t_i < t}{\sum}\alpha_i h\left(t_i,t-t_i\right).
		\end{align}
		As mentioned before, the absorption rate $g(t)$, i.e., the received signal,   of the tumor cells directly affects the healing efficacy of the drug. Hence, we will design the drug delivery system such that $g(t)$ does not fall below a prescribed value. Since $g(t)$ is a function of $h(t_i,t-t_i)$, it is random due to the diffusion of $\Tx$. Therefore, the design of the drug delivery system has to take into account the statistical properties of $g(t)$, which can be obtained from the results of the statistical analysis of $h(t_i,t-t_i)$.	
		
		\subsubsection{Received Signal for MC System}
		
		For the MC system design,  the received signal, denoted by $q_i$, is defined as the number of $\mathrm{X}$ molecules  absorbed at $\Rx$ during  bit interval  $\Tb$  after the transmission of the $i$-th bit   at $t_i$ by $\Tx$ as the received signal, denoted by $q_i$. We detect the transmitted information based on the received signal, $q_i$.  It has been shown in \cite{NCS:14:INB} that $q_i$ follows a Binomial distribution that can be accurately approximated by a Gaussian distribution when $\alpha_i$ is large, which we  assume here.
		We focus on the effect of the transceivers' movements on the MC system performance and  design the optimal release profile of $\Tx$ to account for these movements. We assume  the bit interval to be sufficiently long such that most of the molecules have been captured by or have moved far away from  $\Rx$ before  the following bit is transmitted, i.e.,  ISI is negligible. We note that enzymes \cite{NCS:14:INB} and reactive information molecules, such as acid/base molecules \cite{FG:16:SPAWC, JFSG:18:ICC}, may be used to speed up the molecule removal process and to increase the accuracy of the ISI-free assumption.  Moreover, we  model  external noise sources in the environment as  Gaussian background noise with mean and variance equal to $\eta$ \cite{JAW:18:Arxiv}. Thus, we have 
		\begin{align} \label{eq:51}
		q_i\sim 	\mathcal{N}\left(\mu_{i,\tilde{b}_i},\sigma_{i,\tilde{b}_i}^2\right) \text{ for } b_i=\tilde{b}_i,
		\end{align}
		where $\mu_{i,0}=\sigma_{i,0}^2=\eta$, 
		$\mu_{i,1}=\alpha_i p(t_i,\Tb)+\eta$,  $\sigma_{i,1}^2=\alpha_i p(t_i, \Tb)(1-p(t_i, \Tb))+\eta$. Here, $\mathcal{N}\left(\mu,\sigma^2\right)$ denotes a Gaussian distribution with mean $\mu$ and variance $\sigma^2$.
		 $p(t, \Tb)$ denotes  the  probability that a signaling molecule is absorbed during  bit interval $\Tb$  after its release at time $t$ at the center of $\Tx$.  For  a given distance $r(t)$,   $p(t, \Tb)$ is given by \cite{YHT:14:CL}
		\begin{align} \label{eq:33}
		p(t, \Tb)=\int_0^{\Tb} h(t,\tau)\diff \tau =\frac{\arx}{r(t)}\erfc\left(\frac{r(t)-\arx}{2\sqrt{D_1\Tb}}\right),
		\end{align}
		where $\erfc(\cdot)$ is the complementary error function.
		Since $r(t)$ is a random variable and $p(t, \Tb)$ is a function of $r(t)$, $p(t, \Tb)$ and any function of $p(t, \Tb)$, e.g., the received signal $q_i$, are random processes. Moreover, $p(t, \Tb)$ is also a function of $h(t,\tau)$. 
		 Hence, for MC system design, we have to take into account the statistical properties of $p(t, \Tb)$, which can be obtained based on the proposed statistical analysis of $r(t)$ and $h(t,\tau)$.

		In summary,  the design of both drug delivery  and MC systems depends on   the statistical properties  of the CIR, $h(t,\tau)$, and $r(t)$ including their  means, variances, PDFs, and CDFs, which will be analyzed in the next section.

	\section{Stochastic Channel Analysis} \label{sec:3}

		In this section, we first analyze the distribution of the distance between the transceivers, $r(t)$, and then use it to derive the statistics of the time-variant CIR, $h(t,\tau)$, and $p(t,\Tb)$ as a function $h(t,\tau)$.	In particular,  we develop analytical expressions for the mean, variance, PDF, and CDF of $h(t,\tau)$ and the PDF and CDF of $p(t,\Tb)$.
		
		\subsection{Distribution of the $\Tx$-$\Rx$ Distance for a Diffusive System} \label{sub2:3}
		
		In the 3D  space, $r(t)$ is given by   $r(t)=\sqrt{\sum_{d\in\{x,y,z\}}(r_{d,\Rx}(t)-r_{d,\Tx}(t))^2}$, where $r_{d,\Tx}(t)$ and $r_{d,\Rx}(t)$, $ d\in\{x,y,z\}$, are the Cartesian coordinates representing the positions of  $\Tx$ and  $\Rx$ at time $t$, respectively. 
		Let us assume, without loss of generality, that the diffusion of  $\Tx$ and $\Rx$ starts at $t=0$. Then, given the Brownian motion model for the mobility of  $\Tx$ and $\Rx$, we have  $r_{d,\Tx}(t) \sim \mathcal{N}\left(r_{d,\Tx}(t=0), 2 D_\Tx t\right)$ and $r_{d,\Rx}(t) \sim \mathcal{N}\left(r_{d,\Rx}(t=0), 2 D_\Rx t\right)$, where   we assume that $r_{d,\Tx}(t=0)$ and $r_{d,\Rx}(t=0)$ are known.  Let us define $r_{d}(t)=r_{d,\Rx }(t)-r_{d,\Tx}(t)$. Then, we have 
		$r_{d}(t) \sim \mathcal{N}\left(r_{d}(t=0), 2 D_2 t\right)$, where $D_2=D_\Tx+D_\Rx$ is the effective diffusion coefficient
		capturing the relative motion of  $\Tx$ and  $\Rx$, see
		\cite[Eq.~(10)]{AJN:17:CL}.
		Given the Gaussian distribution of $r_{d}(t)$, we know that  \cite{MBB:58:QAM}
		\begin{align}\label{eq:4}
		\gamma=\frac{r(t)}{\sqrt{2 D_2 t}}=\sqrt{\frac{\sum_{d\in\{x,y,z\}}r_{d}^2(t)}{2 D_2 t}},
		\end{align}
		follows a noncentral chi-distribution, i.e., $\gamma \sim \mathcal{X}_k(\lambda)$, with $k=3$ degrees of freedom   and parameter $\lambda=\sqrt{\frac{\sum_{d\in\{x,y,z\}}r_{d}^2(t=0)}{2 D_2 t}} =\frac{r_0}{\sqrt{2 D_2 t}}$, where $r_0$ denotes $r(t=0)$.  The statistical properties  of random variable $r(t)$  are provided in the following lemma.
		\begin{lemma} \label{lem:0}
			The mean, variance, PDF, and CDF of random variable $r(t)$, which represents the distance between the centers of the diffusive mobile $\Tx$ and $\Rx$, are given by, respectively,
			\begin{align}
			\label{eq:5c}
			\mathrm{E}\left\{r(t)\right\}&=\sqrt{\frac{4 D_2 t}{\pi}} \exp\left(-\frac{r_0^2}{4 D_2 t}\right) +\left(r_0 + \frac{2 D_2 t}{r_0}\right) 
			\erf\left(\frac{r_0}{\sqrt{4D_2 t}}\right),\\
			\label{eq:5d}
			\mathrm{Var}\left\{r(t)\right\}&=r_0^2+6 D_2 t-\mathrm{E}^2\left\{r(t)\right\},\\
			\label{eq:5}
			f_{r(t)}(r)&=\frac{r}{r_0  \sqrt{\pi D_2 t}}\exp\left(-\frac{r^2+r_0^2}{4 D_2 t}\right)\sinh\left(\frac{r_0 r}{2 D_2 t}\right),\\
			\label{eq:29}
			\text{and} \hspace{0.5 cm}	F_{r(t)}(r)&=1-\mathbf{Q}_{\frac{3}{2}}\left(\lambda,\frac{r}{\sqrt{2D_\Tx t}}\right).
			\end{align}
			where $\erf(\cdot)$ is the error function, $\mathbf{Q}_{M}\left(a,b\right)$ is the  Marcum Q-function \cite{Rob:69:Bel}, $\mathrm{E}\left\{\cdot\right\}$ denotes statistical expectation, $\mathrm{Var}\{\cdot\}$ denotes variance,  and $f_{\left\{\cdot \right\}} (\cdot)$ and $F_{\left\{\cdot \right\}} (\cdot)$ denote the PDF and CDF of the random variable in the subscript, respectively.
		\end{lemma}
		\begin{proof}
			Please refer to Appendix~\ref{app:0}.
		\end{proof}	
		
		\begin{remark} \label{re:1a}
			From \eqref{eq:5c} and \eqref{eq:5d}, we can observe that when $t\rightarrow \infty$, we have $\exp\left(-\frac{r_0^2}{4 D_2 t}\right) \rightarrow 1$ and $\erf\left(\frac{r_0}{\sqrt{4D_2 t}}\right)$ $ \rightarrow~0$  and, as a result, $\mathrm{E}\left\{r(t)\right\}\rightarrow \infty$.
			% and $\mathrm{Var}\left\{r(t)\right\}\rightarrow \infty$. 
			Intuitively,  because of diffusion, the transceivers eventually move far away from each other on average.
		\end{remark}
		
		\begin{remark}\label{re:1}
			We note that \eqref{eq:5} was derived under the assumption that  $\Tx$ can diffuse in the entire 3D environment. However, in reality,  $\Tx$ cannot move inside  $\Rx$, i.e., it does not interact with  $\Rx$, and thus will be reflected when it hits  $\Rx$'s boundary. Hence, the actual $f_{r(t)}(r)$, derived in \cite{AJN:17:CL},  differs from  \eqref{eq:5}, e.g., $f_{r(t)}(r)=0$ for $r<\atx+\arx$. However,    for very small $r$, i.e., $r\approx 0$, \eqref{eq:5}  approaches zero. Hence,   \eqref{eq:5} is a valid approximation for the actual $f_{r(t)}(r)$. The validity of this approximation is evaluated in Section~\ref{sec:6} via simulations, where, in our particle-based simulation,
			$\Tx$ is reflected upon collision with $\Rx$ \cite{DNE:15:MBMS}.
			%to its position at the start of the				current simulation step when  hitting the $\Rx$ \cite{DNE:15:MBMS}.
		\end{remark}	
			
		\subsection{Statistical Moments of Time-variant CIR } \label{sub3:1}
			
		In this subsection, we derive the statistical moments of the time-variant CIR, i.e., mean $m(t,\tau)$ and variance $\sigma^2(t,\tau)$. In particular, the mean of the time-variant CIR, $m(t,\tau)$, can be written as
		\begin{align} \label{eq:7}
			m(t,\tau)=\int_0^\infty h(t,\tau)\left|_{r(t)=r}\right.f_{r(t)}(r) \diff r.
		\end{align}
		A closed-form expression for \eqref{eq:7} is provided in the following theorem.
			
		\begin{theorem}
			The mean of the impulse response of a time-variant  channel  with diffusive molecules released by a diffusive transparent transmitter and captured by a diffusive absorbing receiver is given by
			\begin{align} \label{eq:8}
				m(t,\tau)&= \frac{\arx}{4\sqrt{\pi\left(D_1 \tau+D_2 t\right)} r_0 \tau}\exp\left(-\frac{\arx^2}{4D_1 \tau}-\frac{r_0^2}{4D_2 t}\right)\left[-\e^{\frac{v(t,\tau)^2}{4u(t,\tau)}}\left(\frac{v(t,\tau)}{2u(t,\tau)}+\arx\right)\right.\\\nonumber
				&\times\left.\erfc\left(\frac{v(t,\tau)}{2\sqrt{u(t,\tau)}}\right)
				+\e^{\frac{w(t,\tau)^2}{4u(t,\tau)}}\left(\frac{w(t,\tau)}{2u(t,\tau)}+\arx\right)\erfc\left(\frac{w(t,\tau)}{2\sqrt{u(t,\tau)}}\right)\right],
			\end{align}
			where   $u(t,\tau)$, $v(t,\tau)$, and $w(t,\tau)$ are defined, for compactness, as follows
			\begin{align} \label{eq:9}
				u(t,\tau)=\frac{1}{4D_1 \tau}+\frac{1}{4D_2 t},\hspace{5mm} v(t,\tau)=-\frac{\arx}{2D_1 \tau}-\frac{r_0}{2D_2 t} ,\hspace{5mm} w(t,\tau)=-\frac{\arx}{2D_1 \tau}+\frac{r_0}{2D_2 t}.
			\end{align}
		\end{theorem}
			
		\begin{proof}
			Substituting \eqref{eq:6} and \eqref{eq:5}  into \eqref{eq:7} and using the integrals given by \cite[Eq.~(2.3.15.4) and Eq.~(2.3.15.7)]{PBM:86:Book}, we obtain the expression for $m(t,\tau)$  in \eqref{eq:8}.	
		\end{proof} 
			
		\begin{remark}\label{re:2}
			%In \eqref{eq:8}, $m(t,\tau)$ is a combination of exponential and complementary error functions. Thus, its behavior is  more complicated as seen in the numerical results compared to the behavior of the CIR in \eqref{eq:6} in static systems with only exponential function.   Moreover, 
			$m(t,\tau)$ is a function of time $t$. Hence, $h(t,\tau)$ is a non-stationary stochastic process. In general, at large $t$,  $m(t,\tau)$ decreases when $t$ increases and eventually approaches zero when $t~\rightarrow~\infty$. This means that  as $t$ increases, the molecules released by $\Tx$, on average, have a decreasing chance of being absorbed by $\Rx$  since the transceivers move away from each other as mentioned in Remark~\ref{re:1a}.   
		\end{remark}

		In order to obtain the %normalized ACF and the 
		variance of $h(t,\tau)$,
		\begin{align} \label{eq:22}
		\sigma^2(t,\tau)=\phi(t,\tau)-m^2(t,\tau),
		\end{align}
		we first need to find an expression for the second moment $\phi(t,\tau)$, defined as $\phi(t,\tau)=\mathrm{E}\left\{h^2(t,\tau)\right\} $. 
		%Note that $\phi(t_1,t_1)$ cannot be obtained from \eqref{eq:15} by substituting $t_2=t_1$ directly since $p\left(\tilde{t}\right)$  would go to minus infinity and \eqref{eq:15} would be undefined. 
		The following corollary provides an analytical expression for  $\phi(t,\tau)$.
		\begin{corollary}
			%In the limit of $t_2 \rightarrow t_1$, the ACF of $h(t,\tau)$, i.e., 
			$\phi(t,\tau)$ is given by
			\begin{align} \label{eq:18}
			\phi(t,\tau)=&c(t,\tau)\int_0^\infty \left(\exp\left(-\hat{u}(t,\tau)r_1^2-\hat{v}(t,\tau)r_1\right)-\exp\left(-\hat{u}(t,\tau)r_1^2-\hat{w}(t,\tau)r_1\right)\right)\\\nonumber
			&\times\left(r_1-2\arx+\frac{\arx^2}{r_1}\right)\diff r_1, 
			\end{align}
			where 
			\begin{align} \label{eq:19}
			&c(t,\tau)=\frac{\arx^2\e^{-\frac{\arx^2}{2D_1\tau}-\frac{r_0^2}{4D_2 t}}}{8D_1 \pi \tau^3 r_0\sqrt{\pi D_2 t}},\hspace{1cm}& \hat{u}(t,\tau)=\frac{1}{2D_1 \tau}+\frac{1}{4D_2 t},\\\nonumber
			&\hat{v}(t,\tau)=-\frac{\arx}{D_1\tau}-\frac{r_0}{2D_2 t},  &\hat{w}(t,\tau)=-\frac{\arx}{D_1\tau}+\frac{r_0}{2D_2 t}.
			\end{align}	
		\end{corollary}
		\begin{proof}
			From the definition, we have
			\begin{align} \label{eq:20}
			\phi(t,\tau)&=\mathrm{E}\left\{h^2(t,\tau)\right\}=\int_{0}^\infty h^2(t,\tau)\left|_{r(t)=r_1} \right. f_{r(t)}\left(r_1\right)\diff r_1 .
			\end{align}
			Substituting \eqref{eq:6} and \eqref{eq:5} into \eqref{eq:20} and simplifying the expression, we obtain \eqref{eq:18}.
		\end{proof}
		
				\begin{remark}
					The expression in %\eqref{eq:15} and 
					\eqref{eq:18} comprises  integrals of the form  $\int_0^\infty \exp\left(a x^2+b x\right)/x\ \diff x$, where $a$ and $b$ are constants. Such integrals cannot be obtained in closed form. However, the integrals can be evaluated numerically in a straightforward manner.
				\end{remark}

		\subsection{Distribution Functions of the Time-variant CIR } \label{sub3:2}
			
		In this subsection, we derive analytical expressions for the PDF and CDF of $h\left(t,\tau\right)$. The PDF of $h\left(t,\tau\right)$ is  given in the following theorem.
			
		\begin{theorem} \label{theo:3}
			The PDF of the impulse response  of a time-variant  channel  with diffusive molecules released by a diffusive transparent transmitter and captured by a diffusive  absorbing receiver  is given by
			\begin{align} \label{eq:25}
				\begin{cases}
				f_{h\left(t,\tau\right)}(h)=\frac{f_{r(t)}(r_1(h))}{\hat{h}'(r_1(h),\tau)}-\frac{f_{r(t)}(r_2(h))}{\hat{h}'(r_2(h),\tau)}, &\text{for } 0 \leq h < h^\star,\\
				f_{h\left(t,\tau\right)}(h)\rightarrow \infty, &\text{for } h=h^\star,\\ 
				f_{h\left(t,\tau\right)}(h)=0, &\text{otherwise,}
				\end{cases}
			\end{align}
			where $\hat{h}\left(r,\tau\right)$ denotes $h\left(t,\tau\right)$, given by \eqref{eq:6}, as a function of $r(t)$ and $\tau$, $f_{r(t)}(r)$ is given by \eqref{eq:5},   $r_1(h)$ and $r_2(h)$, $r_1(h)< r_2(h)$, are the solutions of the equation $\hat{h}\left(r,\tau\right)=h$, $h^\star$ is the maximum value of $\hat{h}\left(r,\tau\right)$ for all values of $r(t)$,
			and $\hat{h}'(r,\tau) $ is given by
			\begin{align} \label{eq:27}
				\hat{h}'(r,\tau)=&\frac{\arx}{\sqrt{4\pi D_1 \tau^3}}\exp\left(-\frac{\left(r-\arx\right)^2}{4 D_1 \tau}\right) \left(\frac{\arx}{r^2}-\frac{\left(r-\arx\right)}{2D_1 \tau}\left(1-\frac{\arx}{r}\right)\right).
			\end{align}
		\end{theorem}
			
		\begin{proof}
			Please refer to Appendix~\ref{app:1}.
			
		\end{proof}
		
		As stated in the proof of Theorem~\ref{theo:3}, there are two different values of $r(t)$,  $r_1$ and $r_2$, leading to the same value of $\hat{h}(r,\tau)$, i.e.,  $h(t,\tau)$, when $ 0 \leq h < h^\star$. Hence, the PDF of $h(t,\tau)$ is a function of the PDFs of these two  values of $r(t)$. However, when $h(t,\tau)$ reaches its maximum, $f_{h\left(t,\tau\right)}(h)$ approaches infinity and does not depend on $f_{r(t)}(r(h))$ since the probability of $h=h^\star$, i.e., $\Pr(h=h^\star)={f_{h\left(t,\tau\right)}(h)}\ \diff h$, is finite and $\diff h$ approaches $0$ at $h=h^\star$. 
		%$h(r(t),\tau)$ has a definite maximum over all values of $r(t)$  since there is a distance at time $t$ that a  molecule  most likely arrives at the $\Rx$ at time $\tau$ after it is released, not earlier or later.
			
		The CDF of $h(t,\tau)$ is given in the following corollary.
		\begin{corollary} \label{theo:4}
			The CDF of the  impulse response  of a time-variant  channel  with diffusive molecules released by a diffusive transparent transmitter and captured by a diffusive absorbing receiver  is given by
			\begin{align} \label{eq:28}
			\begin{cases}
				F_{h\left(t,\tau\right)}(h)=F_{r(t)}(r_1(h))+1-F_{r(t)}(r_2(h)), &\text{for } 0 \leq h \leq h^\star,\\
				F_{h\left(t,\tau\right)}(h)=0, &\text{for } h<0,\\
				F_{h\left(t,\tau\right)}(h)=1, &\text{for } h>h^\star,
				\end{cases}
			\end{align}
			where $F_{r(t)}(r)$ is  given by \eqref{eq:29}.

		\end{corollary}

		\begin{proof}
			From the definition of the CDF and \eqref{eq:25}, we have
			\begin{align} \label{eq:30}
				F_{h\left(t,\tau\right)}(h)=&\int_0^{h} f_{h\left(t,\tau\right)}(\check{h})\diff \check{h}
				=\int_0^{h} \frac{f_{r(t)}(\check{r}_1(\check{h}))}{ \partial \hat{h}(\check{r}_1,\tau)/\partial \check{r}_1}-\frac{f_{r(t)}(\check{r}_2(\check{h}))}{ \partial \hat{h}(\check{r}_2,\tau)/\partial \check{r}_2}\diff\check{h}\\ \nonumber
				=&\int_0^{{r}_1(h)} f_{r(t)}(\check{r}_1)\diff\check{r}_1-\int_\infty^{{r}_2(h)}f_{r(t)}(\check{r}_2) \diff\check{r}_2
				=F_{r(t)}(r_1(h))+1-F_{r(t)}(r_2(h)),
			\end{align}
			where $\check{r}_1$ and $\check{r}_2$, $\check{r}_1<\check{r}_2$, are the solutions of the equation $\hat{h}\left(\check{r},\tau\right)=\check{h}$.
			This completes the proof.
		\end{proof}

		Similar to the PDF, the CDF of 	$h\left(t,\tau\right)$ also depends on the CDFs of two  values of $r(t)$, i.e., $r_1(h)$ and $r_2(h)$. 
		
		\subsection{Distribution Functions of $p(t,\Tb)$}
		
		Calculating the  mean of $p(t, \Tb)$ involves an integral of the form $\int_0^\infty \erfc\left(ax\right)\exp(-b^2 x^2+c x)\diff x $, with appropriate constants $a, b, c>0$, for which a closed-form expression is not known. However, based on the results in  Subsections~\ref{sub2:3} and \ref{sub3:2}, we  obtain the PDF and CDF of $p(t, \Tb)$  in the following corollary.
		\begin{corollary} \label{cor:4}
			The PDF and CDF of the probability that a diffusive molecule is absorbed by a diffusive absorbing receiver during  an interval $\Tb$  after its release at time $t$   by a diffusive transparent transmitter are, respectively, given by 
			\begin{align}\label{eq:60}
			f_{p(t, \Tb)}(p)=-\frac{f_{r(t)}(\tilde{r}(p))}{p'(\tilde{r})},
			\end{align}		
			\begin{align}\label{eq:61}
			F_{p(t, \Tb)}(p)=1-F_{r(t)}(\tilde{r}(p)),
			\end{align}	
			where $f_{r(t)}(r)$ and $F_{r(t)}(r)$ are  given by  \eqref{eq:5} and \eqref{eq:29}, respectively. Here, $\tilde{r}(p)$ is the solution of the equation $p(t, \Tb)=p$  and $p'(\tilde{r})$ is given by
			\begin{align} \label{eq:62}
			p'(\tilde{r})=&-\frac{\arx}{\tilde{r}^2}\erfc{\left(\frac{\tilde{r}-\arx}{2\sqrt{D_1\Tb}}\right)}-\frac{\arx}{\tilde{r}\sqrt{\pi D_1 \Tb}}\exp\left(-\frac{(\tilde{r}-\arx)^2}{4D_1 \Tb}\right).
			\end{align}
		\end{corollary}
		\begin{proof}
			The proof of Corollary~\ref{cor:4} follows the same steps as  the proof of Theorem~\ref{theo:3} and Corollary~\ref{theo:4} and exploits that $p(t, \Tb)$ is a function of $r(t)$ as shown in \eqref{eq:33}. From \eqref{eq:62}, we observe that $p'(\tilde{r})<0$ so   the equation $p(t, \Tb)=p$ has only one solution. Then, we apply the relations for the PDFs and CDFs of  functions of  random variables \cite{BV:04:Book} to obtain \eqref{eq:60} and \eqref{eq:61}.
		\end{proof}

		The mean, variance, PDF, and CDF of $h\left(t,\tau\right)$ and $p(t,\Tb)$ can be exploited to design  efficient and reliable synthetic MC systems. As examples, we consider the design and analysis of drug delivery and MC systems in the following two sections.

	\section{Drug Delivery System Design}\label{sec:4}
	
	In this section, we apply the derived stochastic parameters of the time-variant  CIR for absorbing receivers for the design and  performance evaluation of drug delivery systems.

		\subsection{Controlled-Release Design} \label{sub4:2}
	
		The treatment of many diseases requires the diseased cells to absorb a minimum rate of drugs during a prescribed time period at minimum cost \cite{KS:14:DD}. To design an efficient drug delivery system satisfying this requirement, we minimize  the total number of released drug molecules, $A=\sum_{i=1}^I \alpha_i$,  subject to the constraint that the absorption rate $g(t)$ is  equal to or larger than a target rate, $\theta(t)$, for a period of time, denoted  by $\TRx$. We allow  $\theta(t)$ to be a function of time so that the designed system can satisfy different treatment requirements over time. 
		Since $g(t)$ is random, we cannot always guarantee $g(t)\geq \theta(t)$. Hence,
		we will design the system based on the first and second order moments of $g(t)$ and use the PDF and CDF of $g(t)$ to evaluate the system performance. In particular, we reformulate the constraint such that the mean of $g(t)$ minus a certain deviation  is equal to or above the threshold $\theta(t)$ during $\TRx$, i.e., $\mathrm{E}\left\{g(t)\right\}-\beta \mathrm{V}\left\{g(t)\right\}\geq \theta(t)$, $0\leq t \leq \TRx$, where  $ \mathrm{V}\{\cdot\}$ denotes    standard deviation and $\beta$ is a coefficient determining how much deviation from the mean is taken into account. Based on \eqref{eq:5a}, the constraint can be written as a function of $\alpha_i$ as follows
		\begin{align} \label{eq:35}
			&\mathrm{E}\left\{g(t)\right\}-\beta\mathrm{V}\left\{g(t)\right\}\overset{(a)}{\geq} \sum_{\forall i | t_i<t}\alpha_i\left(\mathrm{E}\left\{h\left(t_i,t-t_i\right)\right\}-\beta\mathrm{V}\left\{h\left(t_i,t-t_i\right)\right\} \right){\geq} \theta(t),
		\end{align}
		for $0\leq t\leq \TRx$. Inequality $(a)$ in \eqref{eq:35} is due to $\mathrm{E}\left\{g(t)\right\}=\mathrm{E}\bigg\{\underset{\forall i | t_i < t}{\sum}\alpha_i h\left(t_i,t-t_i\right)\bigg\}=\sum_{\forall i | t_i<t}\alpha_i$ \ $\times\mathrm{E}\left\{h\left(t_i,t-t_i\right)\right\}$ and Minkowski's inequality \cite{Ste:08:Lec}:
		\begin{align}
		\mathrm{V}\left\{g(t)\right\}&=\left[\mathrm{E}\left\{\left(\underset{\forall i | t_i < t}{\sum}\left(\alpha_i h\left(t_i,t-t_i\right)-\mathrm{E}\left\{\alpha_i h\left(t_i,t-t_i\right)\right\}\right)\right)^2\right\}\right]^{1/2}\\\nonumber
		&\leq\underset{\forall i | t_i < t}{\sum}\alpha_i \left[\mathrm{E}\left\{\left( h\left(t_i,t-t_i\right)-\mathrm{E}\left\{ h\left(t_i,t-t_i\right)\right\}\right)^2\right\}\right]^{1/2}=\sum_{\forall i | t_i<t}\alpha_i\mathrm{V}\left\{h\left(t_i,t-t_i\right)\right\}.
		\end{align}
	%	Inequality $(b)$ is a choice of constraint.
		
		Note that we may not be able to  find $\alpha_i$ such that \eqref{eq:35}  holds for all values of $\beta$ and $\theta(t)$. For example, when $\beta$ is too large, $\mathrm{E}\left\{g(t)\right\}-\beta\mathrm{V}\left\{g(t)\right\}$ can be negative and hence \eqref{eq:35} cannot be satisfied for $\theta(t)>0$.  However, when $\mathrm{E}\left\{h\left(t_i,t-t_i\right)\right\}>\beta\mathrm{V}\left\{h\left(t_i,t-t_i\right)\right\}$, i.e., either $\beta$  or $\mathrm{V}\left\{h\left(t_i,t-t_i\right)\right\}$ are small, such that $\beta\mathrm{V}\left\{h\left(t_i,t-t_i\right)\right\}$ is sufficiently small, we can always find $\alpha_i$  so that \eqref{eq:35}  holds for arbitrary  $\theta(t)$. Since time $t$ is a continuous variable, the constraint in \eqref{eq:35} has to be satisfied for all values of $t$, $0\leq t \leq \TRx$, and thus there is an infinite number  of constraints, each of which corresponds to one value of $t$. Therefore, we simplify the problem by relaxing the constraints to hold only for  a finite number of time instants $t=t_n=n\Delta t_n$, where $n=1,\dots,N$ and $\Delta t_n=\TRx/N$. Then, the proposed optimization problem for the design of $\alpha_i$ is formulated as follows
%		\begin{subequations}\label{eq:36}
			\begin{align}\label{eq:36}
				\underset{\alpha_i\geq 0, \forall i}{\min} &A=\sum_{i=1}^I \alpha_i \\\nonumber
				\text{s.t. }  &\sum_{i,t_i<t}\alpha_i\left( m\left(t_i,n\Delta t_n-t_i\right)- \beta  \sigma\left(t_i,n\Delta t_n-t_i\right)\right)\geq \theta(n\Delta t_n), \text{ for } n=1,\dots, N, 
			\end{align}
%		\end{subequations}
		where $m\left(t,\tau\right)$  and   $\sigma\left(t,\tau\right)$ are the mean \eqref{eq:8} and the standard deviation \eqref{eq:22} of $h\left(t,\tau\right)$, respectively.
		Since $m(t,
		\tau)$ and $\sigma(t,\tau)$ do not oscillate  but are  well-behaved and smooth functions of $t$ as shown in Section~\ref{sec:6}, a small value of $N$ (e.g., $N=5$) is usually enough to meet  the continuous constraint \eqref{eq:35} for all $t$.
		 Having $m\left(t,\tau\right)$ in \eqref{eq:8} and   $\sigma\left(t,\tau\right)$ in \eqref{eq:22} and treating the $\alpha_i$ as real numbers, \eqref{eq:36} can be readily solved numerically as a linear program. We note that although the numbers of drug molecules $\alpha_i$ are integers, for tractability,  we  solve \eqref{eq:36} for real  $\alpha_i$ and quantize  the results to the nearest integer values.

		We note that the problem in \eqref{eq:36} is statistical in nature and provides  guidance for the offline design of the drug delivery system.

		\subsection{System Performance} \label{sub4:3}

		Since $g(t)\geq \theta(t)$ is required for proper operation of the system, we  evaluate the system performance in terms of the probability that the drug absorption rate satisfies the target rate $g(t)\geq \theta(t)$, denoted by $\Pt(t)=\mathrm{Pr}\left\{g(t)\geq \theta(t)\right\}$. 	
		$\Pt(t)$ is given in the following theorem.
		\begin{theorem}
			The system performance  metric $\Pt(t)$ can be expressed as
			\begin{align}\label{eq:38}
				\Pt(t)=&1-f_{\alpha_1 h\left(t_1,t-t_1\right)}\left(\theta(t)\right)\ast \dots \ast f_{\alpha_{i-1} h\left(t_{i-1},t-t_{i-1}\right)}\left(\theta(t)\right)\ast  F_{\alpha_{i} h\left(t_{i},t-t_{i}\right)}\left(\theta(t)\right),
			\end{align}
			where $\ast$ denotes  convolution, and $i=1,2,\dots$ satisfies $t_{i}\leq t$. In \eqref{eq:38}, we define $f_{\alpha_{i} h\left(t_{i},t-t_{i}\right)}\left(\theta(t)\right)=1/\alpha_{i}\times f_{ h\left(t_{i},t-t_{i}\right)}\left(\theta(t)/\alpha_{i}\right)$ and $F_{\alpha_{i} h\left(t_{i},t-t_{i}\right)}\left(\theta(t)\right)=F_{h\left(t_{i},t-t_{i}\right)}\left(\theta(t)/\alpha_{i}\right)$.
		\end{theorem}
		\begin{proof}
			From the definition of the CDF, we have
			\begin{align}\label{eq:38a}
				&\Pt(t)=1-F_{g(t)}\left\{\theta(t)\right\}=1-\int_0^{\theta(t)}f_{g(t)}(g)\diff g.
			\end{align}
			Due to the summation of independent random variables in \eqref{eq:5a}, i.e., independent releases at $t_i$, we have
			\begin{align}\label{eq:38b}
				f_{g(t)}(g)=\left(f_{\alpha_1 h\left(t_1,t-t_1\right)}\ast \dots \ast f_{\alpha_{i} h\left(t_{i},t-t_{i}\right)}\right)(g).
			\end{align}
			%where $F_{g(t)}\left\{\theta(t)\right\}$ is the CDF of $g(t)$. 
			Substituting \eqref{eq:38b} into \eqref{eq:38a}, then using the integration property of the convolution, i.e.,
			\begin{align}\label{eq:39}
				&\int_0^{\theta(t)}\left(f_{\alpha_1 h\left(t_1,t-t_1\right)}\ast \dots \ast f_{\alpha_{i} h\left(t_{i},t-t_{i}\right)}\right)(g)\diff g=f_{\alpha_1 h\left(t_1,t-t_1\right)}\left(\theta(t)\right)\ast \dots \ast \int_0^{\theta(t)}f_{\alpha_{i} h\left(t_{i},t-t_{i}\right)}(g)\diff g,
			\end{align}
			and using the definition of the CDF, we obtain \eqref{eq:38}. 
		\end{proof}
		
			We note that the analytical expressions for the PDF and CDF of $h(t,\tau)$ in Theorem~\ref{theo:3} and Corollary~\ref{theo:4}, respectively, are not in closed form. Nevertheless, the evaluation of the system performance in \eqref{eq:38} can be approximated by a discrete convolution which can be easily evaluated numerically.

		Furthermore, we note that a minimum value of $\Pt(t)$ can be guaranteed based on the statistical moments of the CIR without  knowledge of the  PDF and the CDF as shown in the following proposition.
		\begin{proposition} \label{prop:1}
			For a given solution of \eqref{eq:35}, a lower bound on $\Pt(t)=\mathrm{Pr}\left\{g(t)\geq \theta(t)\right\}$ is given as follows
			\begin{align}\label{eq:37}
				&\Pt(t)\geq 1-\frac{1}{\beta^2}.
			\end{align}
		\end{proposition}
		\begin{proof}
			By using \eqref{eq:35} and the Chebyshev inequality \cite{PP:02:Book}, we obtain
			\begin{align}\label{eq:37a}
				&\Pt(t)\overset{(a)}{\geq} \mathrm{Pr}\Big\{\left|g(t)-\mathrm{E}\left\{g(t)\right\}\right|\leq \mathrm{E}\left\{g(t)\right\}-\theta(t)\Big\}\overset{(b)}{\geq} 1-\frac{\mathrm{V}^2\left\{g(t)\right\}}{ \left(\mathrm{E}\left\{g(t)\right\}-\theta(t)\right)^2}\overset{(c)}{\geq}  1-\frac{1}{\beta^2},
			\end{align}
			where $(a)$ can be obtained  by expanding the absolute value on the right-hand side, $(b)$ is due to the  Chebyshev inequality, and $(c)$ is due to \eqref{eq:35}. This completes the proof.
		\end{proof}
		\begin{remark} \label{re:3}
			Proposition~\ref{prop:1} is not only useful  for evaluating the system performance, but also provides a guideline for the design of the release profile of drugs in \eqref{eq:36}. For example,  to ensure a high absorption rate probability of $\Pt(t) \geq 0.75$, from \eqref{eq:37}, we need to set the $\beta$ coefficient in \eqref{eq:36} as  $\beta=2$. Note that a useful bound can only be obtained based on \eqref{eq:37} when $\beta>1$ and \eqref{eq:35} is satisfied.
		\end{remark}

		\section{MC System Design for Imperfect CSI}\label{sec:5}

		In this section, we  apply the stochastic  analysis presented in Section~\ref{sec:3} for the  design of MC systems with imperfect CSI. The CSI  is imperfect due to the movement of the transceivers and  assumed to be known only at the beginning of a bit frame.
		In particular, we optimize three design parameters of a diffusive mobile MC system employing on-off keying modulation and threshold detection, namely the detection threshold at  $\Rx$, the release profile at  $\Tx$, and the time duration of a bit frame. By choosing the optimal values of those three parameters, we can improve the system performance while keeping the overall system relatively simple. First, we optimize the detection threshold for  minimization of the maximum BER in a frame assuming a uniform release profile. This approach can be employed  in  very simple MC systems where $\Tx$ is not capable of adjusting the number of released molecules. Second, we optimize the release profile at  $\Tx$  for  minimization of the  maximum BER in a frame, assuming a fixed detection threshold and a fixed  number of molecules available for transmission in the frame. This second approach to MC optimization improves the system performance in terms of BER but  requires a mechanism to control the number of  molecules released at  $\Tx$. Third, we design the optimal duration of the bit frame satisfying a constraint on the efficiency of molecule usage. Thus, this third approach improves the system performance in terms of the efficiency of molecule usage. The three proposed designs can be performed offline. Furthermore, they can be combined with each other or carried out separately depending on the capabilities and requirements of the system. For all three designs, as a first step,  we
		need to derive the BER as a
		function of the number of released molecules.

		\subsection{Detection and BER}
		We consider a simple threshold detector at  $\Rx$, where the received signal $q_i$ is compared with a detection threshold, denoted by $\xi$, in order to determine the detected bit $\hat{b}_i$  as follows   
		\begin{align} \label{eq:50}
		\hat{b}_i=
		\begin{cases}
		1 \text{ if } q_i>\xi,\\
		0 \text{ if } q_i \leq \xi.
		\end{cases}
		\end{align}

		Given the assumption of no ISI and $\Pr(\tilde{b}_i)=1/2$, from \eqref{eq:51} and \eqref{eq:50}, the  error probability of the $i$-th bit, denoted by $\Pe(b_i)$, can be simplified as \cite[Eq.~(12)]{AJN:17:CL} 
		\begin{align}\label{eq:53}
		\Pe(b_i)=\frac{1}{2}-\frac{1}{4}\erf\left(\frac{\xi-\eta}{\sqrt{2\eta}}\right)+\frac{1}{4} \int_0^\infty f_{r(t)}(r_i) \erf\left(\zeta_i(\xi,\alpha_i)\right)\diff r_i,
		\end{align}                                              
		where	 $f_{r(t)}(r_i)$ is given in \eqref{eq:5}, $r_i$ is $r(t_i)$ for brevity,  and $\zeta_i(\xi,\alpha_i)=\frac{\xi-\mu_{i,1}}{\sigma_{i,1}\sqrt{2}}=\frac{\xi-\left(\alpha_i p(t_i, \Tb)+\eta\right)}{\sqrt{2\left(\alpha_i p(t_i, \Tb)(1-p(t_i, \Tb))+\eta\right)}}$. Note that $\Pe(b_i)$ depends on $i$ since the distance $r(t_i)$ between the transceivers is a function of release time $t_i$. 
		
		\subsection{Optimal Detection Threshold for Uniform Release}\label{sub5:2}
		
		We first consider system  design for uniform release, where the  number of available molecules is uniformly allocated across all bits of a frame. To facilitate reliable communication, our objective is to optimize the detection threshold, $\xi$, such that   the maximum value of the error rate of the bits in a frame is minimized, given the total number of available molecules in a frame, $A$, i.e.,
		\begin{align}
			\min_\xi\max_i \left\{\Pe(b_i)\right\}\hspace{9mm} \mathrm{s.t.}\hspace{2mm} \alpha_i=A/I.
		\end{align}  
		From \eqref{eq:53}, the problem is equivalent to
		\begin{align} \label{eq:58}
		\min_{\xi}\max_i\left\{ \int_0^\infty f_{r(t)}(r_i) \erf\left(\zeta_i(\xi,\alpha_i)\right)\diff r_i-\erf\left(\frac{\xi-\eta}{\sqrt{2\eta}}\right)\right\}\hspace{9mm} \mathrm{s.t.}\hspace{2mm} \alpha_i=A/I.
		\end{align}
		 The following lemma reveals the convexity of the problem in \eqref{eq:58}. 
		\begin{lemma} \label{lem:2}
			For $\eta<\xi<\mu_{i,1}$, the objective function in   \eqref{eq:58} is  convex in $\xi$.
		\end{lemma}
		\begin{proof}
			Please refer to Appendix~\ref{app:3}.
		\end{proof}
		Note that  $\eta<\xi<\mu_{i,1}$ is intuitively satisfied for typical system parameters  since the decision threshold should be higher than the average noise level  when bit "$0$" is sent but should not exceed the mean of the received signal when bit "$1$" is sent. Otherwise, a high error rate would result.
	    Due to the convexity of problem \eqref{eq:58}, the global optimum  $\xi$ can be easily obtained by  numerical methods such as the interior-point method \cite{BV:04:Book}.
		
			\subsection{Optimal Release  with Fixed Detection Threshold}\label{sub5:2a}
			For the second proposed design, we aim to optimize the release profile, i.e., the number of  molecules available for release for each bit, $\alpha_i$, such that   $\max_i \left\{\Pe(b_i)\right\}$ is minimized given a total number of  molecules $A$ that are available for release in a frame
			\begin{align} \label{eq:55}
			\min_{\bm{\alpha}} \max_i \left\{\Pe(b_i)\right\} \hspace{9mm} \mathrm{s.t.}\hspace{2mm} \sum_{i=1}^I \alpha_i=A,
			\end{align}
			where 	  $\bm{\alpha}=\left[\alpha_i,\alpha_2,\dots,\alpha_I\right]$.
			
			For a given threshold $\xi$,  we can re-express \eqref{eq:55} based on \eqref{eq:53} as
			\begin{align} \label{eq:55a}
			\min_{\bm{\alpha}} \max_i \left\{ \int_0^\infty f_{r(t)}(r_i) \erf\left(\zeta_i\left(\xi,\alpha_i\right)\right)\diff r_i \right\}\hspace{9mm} \mathrm{s.t.} \hspace{2mm} \sum_{i=1}^I \alpha_i=A.
			\end{align}
			
			The following lemma states the convexity of the optimization problem in \eqref{eq:55a}.  
			\begin{lemma} \label{lem:1}
				For $\eta<\xi<\mu_{i,1}$, the objective function in   \eqref{eq:55a} is  convex in $\bm{\alpha}$.
			\end{lemma}
			\begin{proof}
				Please refer to Appendix~\ref{app:2}.
			\end{proof}
					
			 	Hence, the global optimum of \eqref{eq:55a} can be readily obtained by numerical methods such as the interior-point method. 
			 	
			 	Note that, for tractability, similar to the proposed drug delivery  design, we solve \eqref{eq:58} and \eqref{eq:55a} for real $\alpha_i$ and quantize the results to the nearest integer values.

		\subsection{Optimal Time Duration of a Bit Frame} \label{sub5:3}
		
		In the third proposed design, we consider the  molecule usage efficiency for communication. We evaluate the efficiency based on  $p(t, \Tb)$, i.e., the probability that a signaling molecule is absorbed  during  bit interval  $\Tb$ after its  release at time $t$.	If $p(t, \Tb)$ is too small, none of the  released molecules may actually reach the receiver and thus the molecules are wasted, i.e., the system has low efficiency. Hence, we want to keep  $p(t, \Tb)$ above a certain value, denoted by $\psi$.	Intuitively,  as on average $h(t,\tau)$  decreases over time $t$, $p(t, \Tb)$, which is the integral over $h(t,\tau)$ with respect to $\tau$, also on average decreases over time. Therefore, our objective is to choose the maximum duration of a bit frame, denoted by $T^\star$, such that $p(t, \Tb)>\psi$ for $t\leq T^\star-\Tb$, where $t$ is the release time.  
		
		Since   $p(t, \Tb)$ is a random process, we cannot enforce $p(t, \Tb)>\psi$ but can only bound the probability that $p(t, \Tb)>\psi$ is satisfied, i.e.,  $\Pr\left(p(t, \Tb)>\psi\right)\geq P$, where $P$ is a design parameter.	
		 Moreover, we have 
		\begin{align}
		\Pr\left(p(t, \Tb)>\psi\right)=1-F_{p(t, \Tb)}(\psi)\overset{(a)}{=}F_{r(t)}(\tilde{r}(\psi)),
		\end{align}
		where equality $(a)$ is due to \eqref{eq:61}. As such, we can re-express the problem as maximizing the duration of a bit frame such that $F_{r(t)}(\tilde{r}(\psi))\geq P$ holds. To this end, in the following lemma, we analyze $F_{r(t)}(\tilde{r}(\psi))$ as a function of time $t$. 
		
			\begin{lemma}\label{lem:3}
				$F_{r(t)}(\tilde{r}(\psi))$ is a decreasing function of time $t$.
			\end{lemma}
			\begin{proof}
				Please refer to Appendix~\ref{app:4}.
			\end{proof}
			
			Since Lemma~\ref{lem:3} shows that $F_{r(t)}(\tilde{r}(\psi))$ is a decreasing function of time, the maximum duration of a bit frame satisfying $F_{r(t)}(\tilde{r}(\psi))\geq P$ can be found by
			 solving $F_{r(T^\star-\Tb)}(\tilde{r}(\psi))= P$, where $F_{r(T^\star-\Tb)}(\tilde{r}(\psi))$ is given   in \eqref{eq:29}.

		\begin{remark}
			If multiple frames are transmitted, the proposed design framework  can be applied to 
			each frame, respectively. However, the optimal designs may be different for different frames due to the moving transceivers, whose distances are assumed to be perfectly estimated at the start of each frame.
		\end{remark}

	\begin{remark}
		Here, we discuss a system with an absorbing receiver. Nevertheless,
		the proposed  optimal design framework can also be applied to  transparent receivers. For a transparent receiver, $p(t, \Tb)$ is the  probability that a molecule is observed inside the volume 	of the transparent receiver at time $\Tb$ after its release at time $t$ at the center of $\Tx$.
	\end{remark}

	\section{Numerical Results} \label{sec:6}
	
	In this section, we provide  numerical results to evaluate the accuracy of the derived expressions and analyze the performance of the MC systems in the considered application scenarios. We use the set of simulation parameters summarized  in Table~\ref{tab:1}, unless stated otherwise. The parameters are chosen to match the actual system parameters in drug delivery systems, as will be explained in detail in Subsection~\ref{sub7:2}. 
	% For the simulations of the MC system design, we change some of the parameters to generally describe a MC system, where both transceivers diffuse and the communication time is shorter. 

	\begin{table}[t!]
		\captionof{table}{System parameters used for numerical results}
		\centering
		\label{tab:1}
		\begin{tabular}{c|c||c|c}
			\toprule
			%    		& \multicolumn{5}{c}{$R$ [bits/s]} \\
			%    		\cline{2-6} \\[-0.2cm]
			Parameter	&	Value	&	Parameter	&	Value \\[-0.05cm]
			\midrule
			$D_\Tx$ [$\mathrm{m^2/s}$]	&	$1\times 10^{-14}$		&	$D_\Rx$ [$\mathrm{m^2/s}$]	&	$0$ 	\\
			$D_\mathrm{X}$ [$\mathrm{m^2/s}$]	&	$8\times 10^{-11}$		&	$r_0$ [$\mathrm{m}$]	&	$10\times 10^{-6}$ 	\\
			$\atx$ [$\mathrm{m}$]	&	$1\times 10^{-7}$&	$\arx$ [$\mathrm{m}$]	&	$1\times 10^{-6}$		\\
			$T$ [$\si{\hour}$]	&	$24$& $\TRx$ [$\si{\hour}$]	 &	$24$	\\
			$I$	&	$3000$	&	$N$	&	$5$\\
			$\Tb$[$\si{\second}$] & $28.8$& $\theta(t) [\si{\second^{-1}}]$ & $1$\\
			\bottomrule
		\end{tabular}
	\end{table}
	
	\subsection{Time-variant Channel Analysis}\label{sub7:1}

	\begin{figure*}[!tbp]
		\centering
		\begin{minipage}[t]{0.49\textwidth}\hspace*{-5 mm}
			\centering
			\resizebox{1.05\linewidth}{!}{
				\includegraphics[scale=0.55]{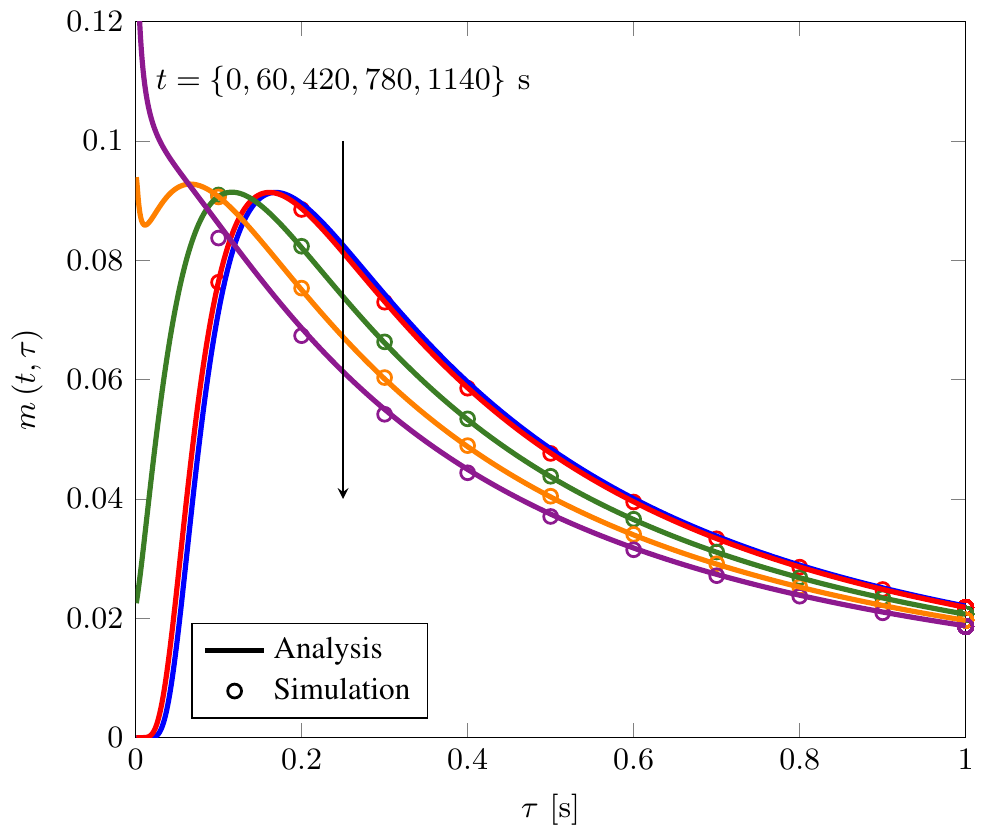}}\vspace*{-4 mm}
			\caption{
				Mean of the CIR $h(r(t),\tau)$ as	a function of time $\tau$.} %The inset shows the CIR $h(r(t),\tau)$ as a function of $r(t)$.	}
			\label{fig:3}
		\end{minipage}
		\hfill
		\vspace*{-1 mm}
		\begin{minipage}[t]{0.49\textwidth}
			\centering
			\resizebox{1.05\linewidth}{!}{
				\includegraphics[scale=0.55]{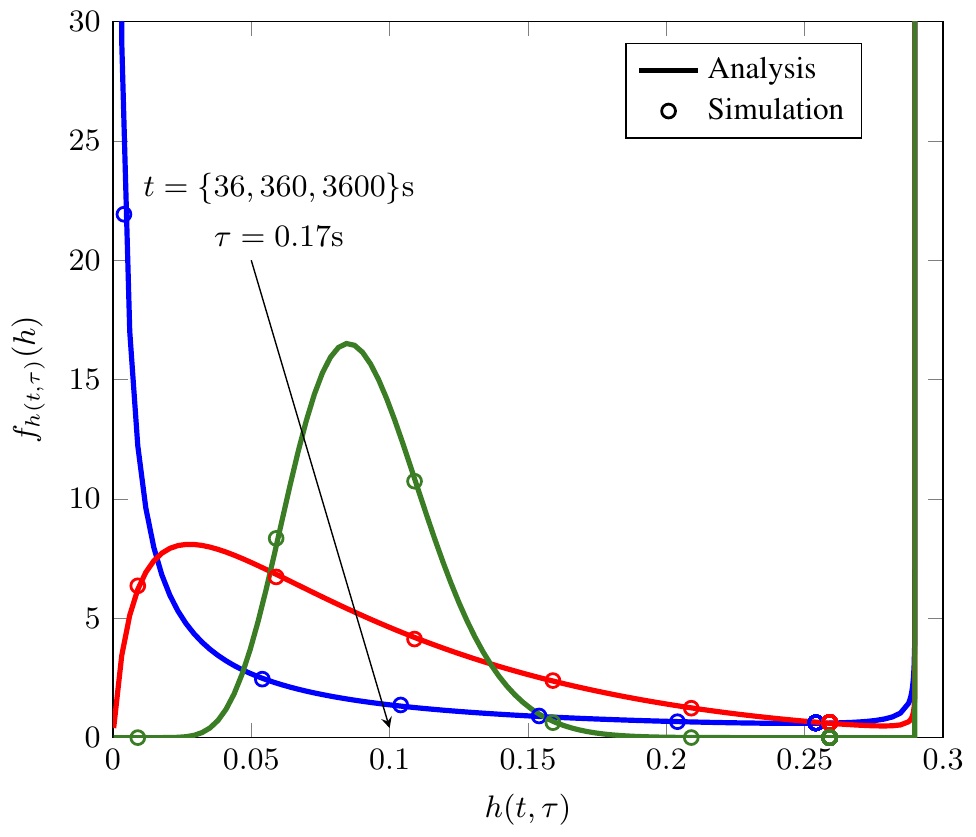}}\vspace*{-4 mm}
			\caption{
				PDF of the CIR $f_{h(t,\tau)}(h)$ for $\tau=\SI{0.17}{\second}$ and $t=\{36,360,3600\}$ \si{\second}. 
			}
			\label{fig:6}
		\end{minipage}
		\vspace*{-1 mm}
		
	\end{figure*}
	%%%%%%%%%%%%%%%%%%%%%%%%%%%%%%%%%%%%%%%%%%%%%%%%%%%%%%%%%%%

		In this subsection, we numerically analyze the time-variant MC channel. For verification of the accuracy of the  expressions derived in Section~\ref{sec:3}, we employ a hybrid particle-based simulation approach. In particular, we use particle-based simulation of the Brownian motion of the transceivers to generate realizations of the random distance between $\Tx$ and $\Rx$, $r(t)$. Then, we use Monte Carlo simulation to obtain the desired statistical results by suitably  averaging the CIRs   \eqref{eq:6} obtained for the realizations of $r(t)$. For particle-based simulation of the Brownian motion of  $\Tx$,    $\Tx$  performs a random walk with a random step size in space in every discrete time step of length $\Delta t^{\mathrm{st}}=\SI{1}{\milli\second}$. The length of  each step in space is modeled as a Gaussian random variable with zero mean and standard deviation  $\sqrt{2D_\Tx\Delta t^{\mathrm{st}}}$. Furthermore, we also take into account the reflection of   $\Tx$ upon collision with $\Rx$. When  $\Tx$ hits  $\Rx$, we assume  that it bounces back to the position it had at the beginning of the
		considered simulation step \cite{DNE:15:MBMS}. 
		%which was verified in \cite{YHT:14:CL} and \cite{AJN:17:CL}

		Fig.~\ref{fig:3} shows the mean of the CIR, $m(t,\tau)$, as a function of $\tau$. 	In general, for large $\tau$,  $m(t,\tau)$ decreases when $t$ increases  as expected since the transceivers move away from each other on average.  For large $\tau$, $m(t,\tau)$ also decreases when $\tau$ increases as would be the case in a static system.
		%Interestingly, the pulse shape does not become wider and shorter when $t$ increases as we would expect to receive less molecules due to the increasing distance between the transceivers on average. In fact, the peak of the pulse slightly increases and appears earlier in terms of time $\tau$. The reason of this behavior can be explained as follows based on the supporting results in the insets of Fig.~\ref{fig:3}. In the upper inset of Fig.~\ref{fig:3}, $\tau_0=\frac{(r_0-\arx)^2}{6D_1}=\SI{0.17}{\second}$ is the time when  $h(t=0,\tau)$ is maximum.
		 %We  observe from the upper inset of Fig.~\ref{fig:3} that $h(t,\tau)$  is much larger for small $r(t)<10^{-5}$ \si{\meter}, and small $\tau$, i.e., a molecule is more likely to arrive in a shorter time  if the distance is small.  
		%From 	the lower inset, we see that the PDF of small $r(t)<10^{-5}$  is higher for larger $t$. Therefore,  $h(t,\tau)$ is much larger on average for larger $t$, i.e. the peak of the pulse,  $m(t,\tau)$, rises and appears earlier.
	 Note that in the  simulations, unlike the analysis, we have taken into account the reflection of  $\Tx$ when it hits  $\Rx$. Therefore, the good agreement between simulation and analytical results in Fig.~\ref{fig:3} suggests that the reflection of  $\Tx$ does not have a significant impact on the statistical properties of $h(t,\tau)$ and the approximation in \eqref{eq:5} and the analytical results obtained based on it are valid.

		In Fig.~\ref{fig:6}, we plot the PDF of the CIR for time instances $t=\left\{36,360,3600\right\}\si{\second}$ and $\tau=\SI{0.17}{\second}$.  Fig.~\ref{fig:6} shows that when $t$ increases, a smaller value of $h(t,\tau)$ is more likely to occur since, on average, the transceivers move away from each other. When $t$ is very large,  it is  likely that the molecules cannot reach $\Rx$, and hence, cannot be absorbed, consequently $h(t,\tau)\rightarrow 0$. We also observe that $h(t,\tau)$ has a maximum value and $f_{h(t,\tau)}\left(h(t,\tau)\right)\rightarrow\infty$ when $h(t,\tau)$ approaches the maximum value as stated in \eqref{eq:25}. For example,  $h(t=\SI{36}{\second},\tau=\SI{0.17}{\second})$ is random but its maximum possible value  is  $0.29$ and $f_{h(t=\SI{36}{\second},\tau=\SI{0.17}{\second})}\left(0.29\right)\rightarrow\infty$.
		
		Fig.~\ref{fig:3} and \ref{fig:6} show a perfect match between simulation and analytical results. This confirms the accuracy of our analysis of the time-variant CIR in Sections~\ref{sec:3}. Since particle-based simulation is costly, in the following subsections,  we adopt Monte-Carlo simulation by averaging our results  over  $10^5$ independent realizations of both the distance $r(t)$ and the CIR. The distance $r(t)$ is calculated from the locations of the transceivers, which are generated from Gaussian distributions, see Subsection~$\ref{sub2:3}$. In particular, $r(t)=\sqrt{\sum_{d\in\{x,y,z\}}(r_{d,\Rx}(t)-r_{d,\Tx}(t))^2}$, where $r_{d,\Tx}(t) \sim \mathcal{N}\left(r_{d,\Tx}(t=0), 2 D_\Tx t\right)$, $r_{d,\Rx}(t) \sim \mathcal{N}\left(r_{d,\Rx}(t=0), 2 D_\Rx t\right)$,   $r_{d,\Tx}(t=0)=0$, and $r_{d,\Rx}(t=0)=1/\sqrt{3}r_0$. The CIR is given by \eqref{eq:6} for each realization of $r(t)$.

	\subsection{Drug Delivery System Design}\label{sub7:2}
	
		In this section, we provide  numerical results  for the considered drug delivery system.  
	As mentioned above, the parameters in Table~\ref{tab:1} are chosen to match real system parameters, e.g., the diffusion coefficient $D_\mathrm{X}$ of drug molecules vary from $10^{-9}$ to $10^{-14}$ \si{\metre^2\per \second} \cite{WCX:14:PCL},  drug carriers have  sizes $\atx \geq \SI{100}{\nano\meter}$ \cite{PNJ:99:BPJ}, the size of tumor cells is on the order of $\si{\micro\metre}$, and  drug carriers can be injected or extravasated from the cardiovascular system into the tissue surrounding the targeted diseased cell site \cite{LEELYP:15:CES}, i.e., close to  the tumor cells. The dosing periods in drug delivery systems are on the order of days \cite{ARI:06:ADD}, i.e., $\SI{24}{\hour}$. For simplicity, we set  $N=5$ and  the value of the required absorption rate is set to $\theta(t)=\SI{1}{\second^{-1}}$.   We choose $I$   relatively large to obtain small intervals $\Tb$.

		%%%%%%%%%%%%%%%%%%%%%%%%%%%%%%%%%%%%%%%%%
		\begin{figure*}[!tbp]
			\centering
			\begin{minipage}[t]{0.49\textwidth}\hspace*{-5 mm}
				\centering
				\resizebox{1.05\linewidth}{!}{
					\includegraphics[width=0.5\textwidth]{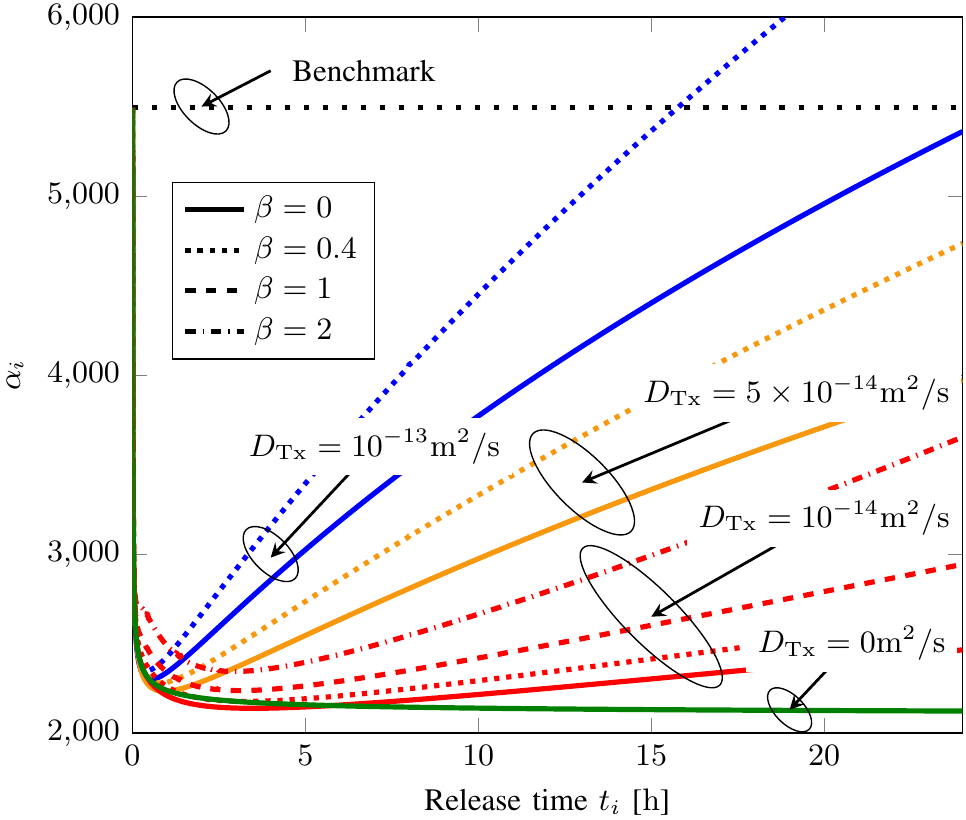}}\vspace*{-4 mm}
				\caption{
					Optimal number of released molecules $\alpha_i$ as a function of release time $t_i$ [\si{\hour}] for different system parameters and $T=\SI{24}{\hour}$. The black horizontal dotted line is the benchmark when the $\alpha_i$ are not optimized. 
				}
				\label{fig:7}
				
			\end{minipage}
			\hfill
			\vspace*{-1 mm}
			\begin{minipage}[t]{0.49\textwidth}\hspace*{-5 mm}
				\centering
				\resizebox{1.05\linewidth}{!}{
					\includegraphics[scale=0.55]{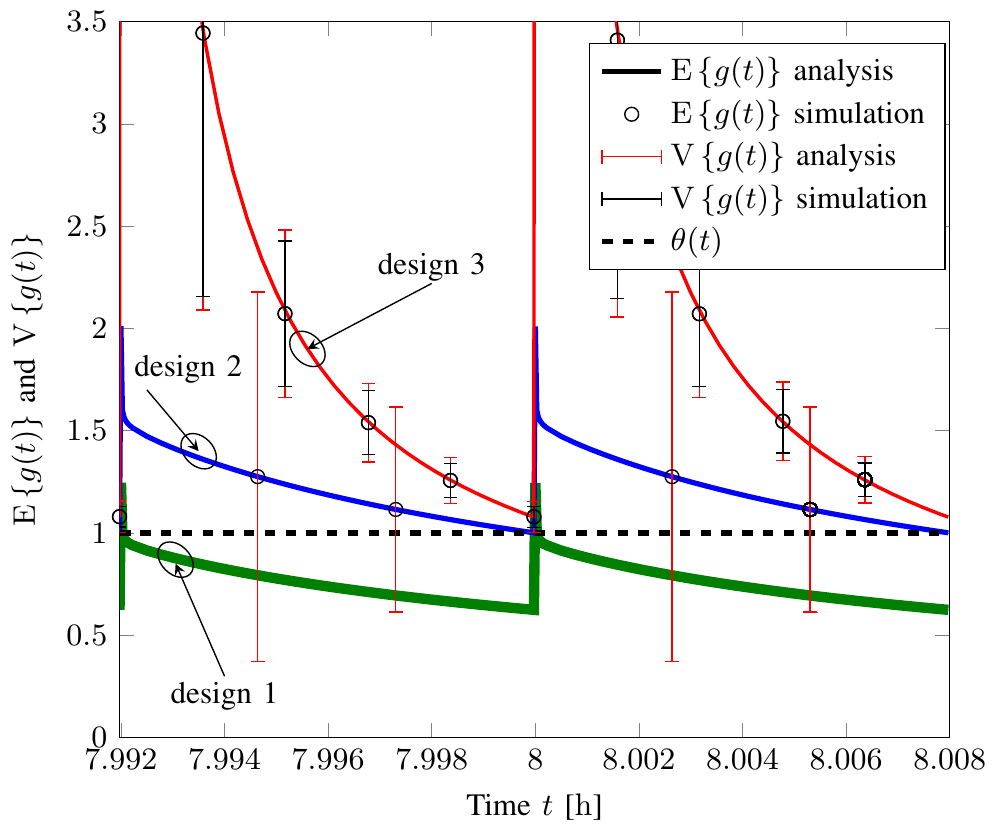}}\vspace*{-4 mm}
				\caption{
					$\mathrm{E}\left\{g(t)\right\}$ and $\mathrm{V}\left\{g(t)\right\}$ between the 1000-th release and the 1002-th release, i.e., at about \SI{8}{\hour}, for three different designs. Design 1 (green line): naive design without considering $\Tx$'s movement  with $D_\Tx=10^{-13}$ \si{\metre^2\per \second} and $\beta=0$; design 2 (blue line) and 3 (red line): optimal design for $\left(D_\Tx[\si{\metre^2\per \second}],\beta\right)=\left(10^{-13},0\right)$, and $\left(10^{-14},1\right)$, respectively.
				}
				\label{fig:8}
			\end{minipage}
		\end{figure*}
		%%%%%%%%%%%%%%%%%%%%%%%%%%%%%%%%%%%%%%%%%%%%%%%%%%%%%%%%%%%

		In Fig.~\ref{fig:7}, we plot the  number of released molecules $\alpha_i$ versus the corresponding release time $t_i$ [$\si{\hour}$] for different system parameters. The  coefficients are obtained by solving the optimization problem in \eqref{eq:36} with $\beta=\{0,0.4,1,2\}$ for $D_\Tx=10^{-14}$ \si{\metre^2\per \second} and $\beta=\{0,0.4\}$ for $D_\Tx=\{5,10\}\times10^{-14}$ \si{\metre^2\per \second}. As mentioned  in the discussion of \eqref{eq:35}, we cannot  choose large values of $\beta$ when the diffusion coefficient is large, i.e., the standard deviation is large, as the problem in \eqref{eq:36} may become infeasible. Fig.~\ref{fig:7} shows that for all considered parameter settings, we  first have to   release a large number of molecules for the absorption rate to exceed the threshold. Then, in the static system with  $D_\Tx= \SI{0}{\metre^2\per \second}$, the optimal coefficient  decreases with increasing time, since a fraction of the molecules previously released   from  $\Tx$ linger around  $\Rx$ and are absorbed later. However, for the  time-variant channel,  $\Tx$ eventually diffuses away from  $\Rx$ as time  $t$ increases and hence,  molecules released at later times by  $\Tx$  will be far away from  $\Rx$ and may not reach it. Therefore, at later times, the  amount of drugs released has to be increased  for the absorption rate to not fall below the threshold. For larger $D_\Tx$,  $\Tx$ diffuses away from  $\Rx$ faster and thus, the number of released molecules $\alpha_i$ have to increase faster. This type of drug release, i.e., first releasing a large amount of drugs, then reducing and eventually increasing the amount of released drugs again, is called a tri-phasic release \cite{FWR:11:JP}. Once we have designed the release profile, we can implement it by choosing a suitable drug carrier as shown in \cite{FWR:11:JP}.  Moreover, as expected, for larger $\beta$, we need to release more drugs to ensure  that \eqref{eq:36} is feasible. The black horizontal dotted line in Fig.~\ref{fig:7} is a benchmark where the $\alpha_i, \forall i,$ are not optimized but naively set to $\alpha_i=\alpha_1=5493$. For this naive design, $A=\alpha_1 I\approx 1.65\times 10^7$, whereas with the optimal $\alpha_i$, for $\beta=0$ and $D_\Tx=10^{-13}$ \si{\metre^2\per \second}, $A=1.2 \times 10^7$, i.e.,  $27\%$ less than the  $A$ required for the naive design, and for $\beta=1$ and $D_\Tx=10^{-14}$ \si{\metre^2\per \second}, $A=7.6 \times 10^6$, i.e.,  $54\%$ less than the  $A$ required for the naive design. This highlights that applying the optimal release profile can save significant amounts of drugs and still satisfy the therapeutic requirements.  Moreover, as observed in Fig.~\ref{fig:7}, at later times, e.g., $t_i> \SI{15}{\hour}$ for $D_\Tx={10^{-13}}\si{\metre^2\per \second}$, the  values of $\alpha_i$ required to satisfy the desired absorption rate are higher than the fixed $\alpha_i$ used in the naive design, i.e., the benchmark, which means that the naive design  cannot provide the required absorption rate.
		%%%%%%%%%%%%%%%%%%%%%%%%%%%%%%%%%%%%%%%%%
		\begin{figure*}[!tbp]
			\centering
			\begin{minipage}[t]{0.49\textwidth}\hspace*{-5 mm}
				\centering
				\resizebox{1.05\linewidth}{!}{
					\includegraphics[scale=0.55]{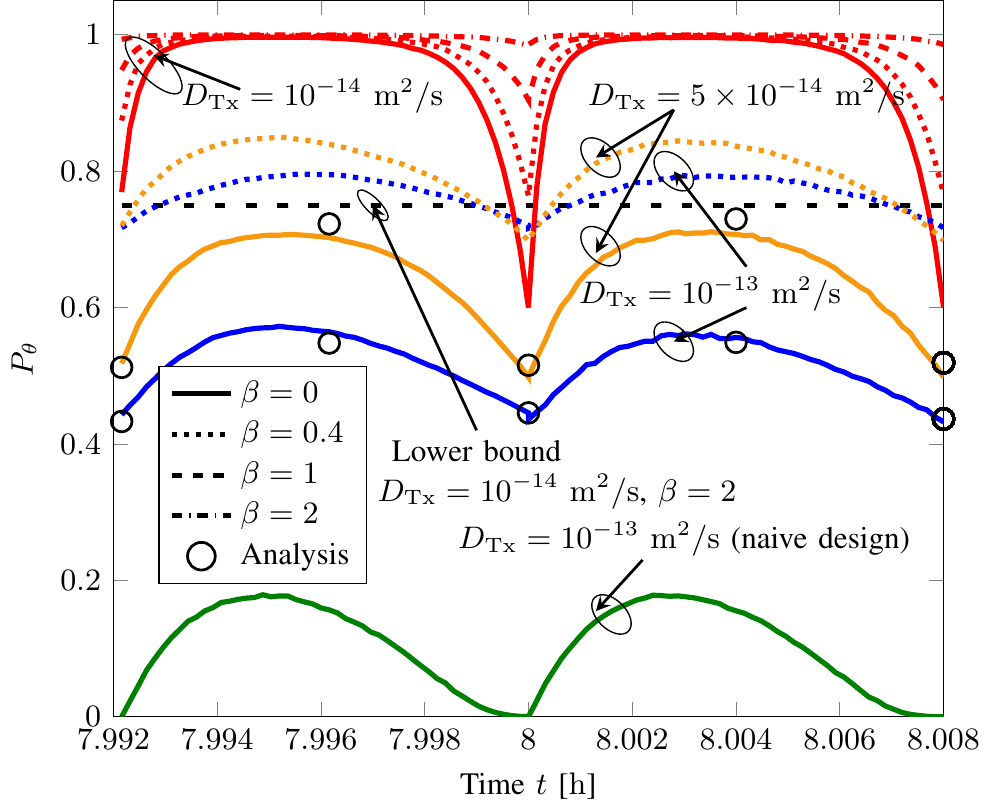}}\vspace*{-4 mm}
				\caption{
					$\Pt(t)$  as a function of time $t$ [$\si{\hour}$]  between the 1000-th release and the 1002-th release, i.e., at about $\SI{8}{\hour}$. 
				}
				\label{fig:9}
			\end{minipage}
			\hfill
			\vspace*{-1 mm}
			\begin{minipage}[t]{0.49\textwidth}\hspace*{-5 mm}
				\centering
				\resizebox{1.05\linewidth}{!}{
				\includegraphics[scale=0.55]{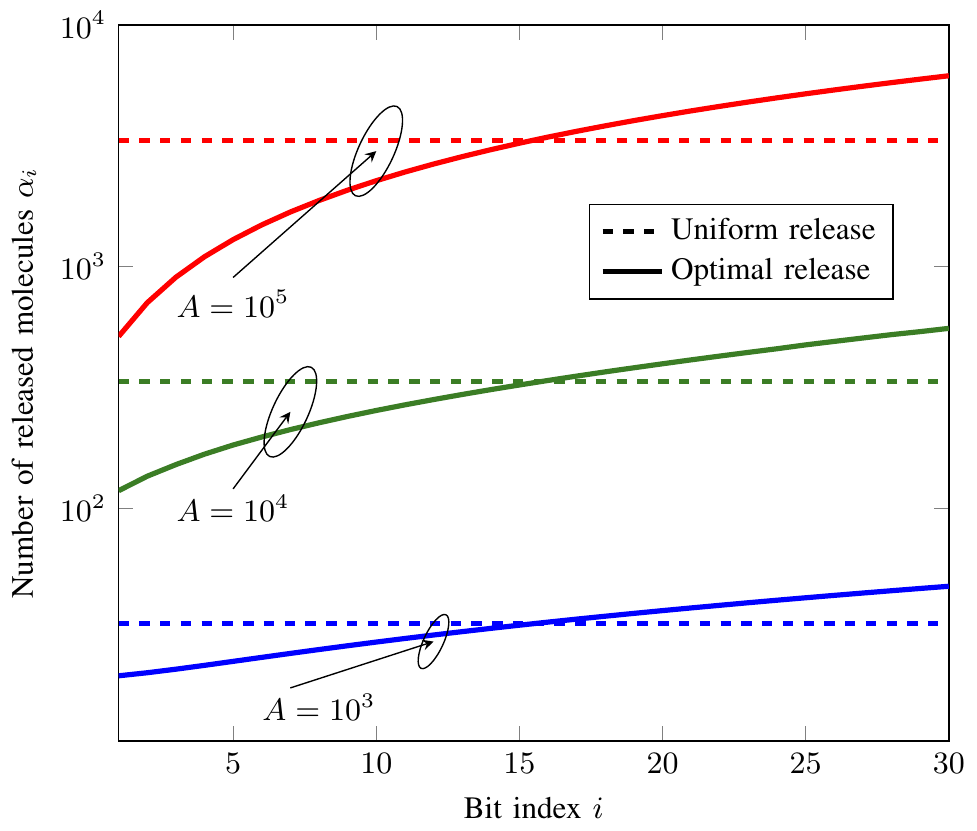}}\vspace*{-4 mm}
			\caption{
				The number of molecules available for release  for each bit for  uniform and optimal release   when $A=\left\{10^3,10^4,10^5\right\}$ and $T=\SI{300}{\second}$.
			}	
			\label{fig:10}	
			\end{minipage}
		\end{figure*}
		%%%%%%%%%%%%%%%%%%%%%%%%%%%%%%%%%%%%%%%%%%%%%%%%%%%%%%%%%%%

		In Fig.~\ref{fig:8}, we plot the mean and standard deviation of the absorption rate,	$\mathrm{E}\left\{g(t)\right\}$ and $\mathrm{V}\left\{g(t)\right\}$, between the  $1000$-th release and the $1002$-th release for three designs. For designs 1 and 2, we assumed $D_\Tx=10^{-13}$ \si{\metre^2\per \second} and $\beta=0$, and for design 3, we adopted  $D_\Tx=10^{-14}$ \si{\metre^2\per \second} and $\beta=1$. Note that the considered time window, e.g., between the  $1000$-th release and the $1002$-th release, is chosen arbitrarily in the middle of $T$ to analyze the system behavior between individual releases. For design 1,  $\Tx$ diffuses with $D_\Tx=10^{-13}$ \si{\metre^2\per \second} but the  release profile is  designed without accounting for  $\Tx$'s mobility, i.e., the adopted $\alpha_i$ are given by the green line in Fig.~\ref{fig:7} obtained under the assumption of $D_\Tx= \SI{0}{\metre^2\per \second}$. For designs 2 and 3,  the mobility of  $\Tx$ is taken into account. The black dashed line marks the threshold $\theta(t)$ that $g(t)$ should not fall below. It is observed from Fig.~\ref{fig:8} that   when  $\Tx$ diffuses but the design does not take into account the mobility,	 the  requirement that the expected  absorption rate,  $\mathrm{E}\left\{g(t)\right\}$, exceeds $\theta(t)$, is not satisfied for most of the time. For design 2 with $\beta=0$, we observe that $\mathrm{E}\left\{g(t)\right\}>\theta(t)$ always holds but  $\mathrm{E}\left\{g(t)\right\}-\mathrm{V}\left\{g(t)\right\}>\theta(t)$ does not always hold. For design 3 with $\beta=1$, we  observe  that $\mathrm{E}\left\{g(t)\right\}-\mathrm{V}\left\{g(t)\right\}>\theta(t)$  always holds since $\beta>0$ enforces a gap between $\mathrm{E}\left\{g(t)\right\}$ and $\theta(t)$. In other words, even if $g(t)$ deviates from the mean, it can still exceed $\theta(t)$. %Fig.~\ref{fig:8} also shows that $\mathrm{E}\left\{g(t)\right\}$  first increases after a release and then decreases, due to the diffusion of the molecules.  
	
		In Fig.~\ref{fig:9}, we present the system performance in terms of  the probability that  $g(t)\geq \theta(t)$, $\Pt(t)$,  for the time period between the $1000$-th and $1002$-th releases, i.e., at about $\SI{8}{\hour}$. The lines and markers denote  simulation and analytical results, respectively. Fig.~\ref{fig:9} shows a good agreement between  analytical and simulation results.  In Fig.~\ref{fig:9}, we observe that  $\Pt(t)$ increases with increasing $\beta$ because the design for larger $\beta$ enforces a larger gap between $\mathrm{E}\left\{g(t)\right\}$ and $\theta(t)$, as can be seen in Fig.~\ref{fig:8}. Moreover, for a given $\beta$, $\Pt(t)$ will be different for different $D_\Tx$. In particular, for larger $D_\Tx$, $\Pt(t)$ is smaller due to the faster diffusion and increasing randomness of the CIR. Moreover, in Fig.~\ref{fig:9}, the green line shows that  the naive design, i.e., design 1 in Fig.~\ref{fig:8}, has very poor  performance. In  Fig.~\ref{fig:9}, we also observe that between two releases, $\Pt(t)$ first increases due to the released drugs and then decreases due to drug diffusion.
		Furthermore, in Fig.~\ref{fig:9}, we also show the lower bound on $\Pt(t)$ derived in Proposition~\ref{prop:1} for $D_\Tx=10^{-14}$ \si{\metre^2\per \second} and $\beta=2$, where \eqref{eq:37} yields $\Pt(t)\geq 0.75$. Fig.~\ref{fig:9} shows that the red dash-dotted line, i.e., $\Pt(t)$ for $D_\Tx=10^{-14}$ \si{\metre^2\per \second} and $\beta=2$, is indeed above the horizontal black dashed line, i.e., $\Pt(t)=0.75$.

	\subsection{Molecular Communication System Design}

		In this subsection, we show  numerical results for the second application scenario, i.e., an MC system with imperfect CSI. We apply again the  system parameters in Table~\ref{tab:1} except that here we set $D_\Rx=10^{-11}~\si{\meter^2/\second}$,  $I=30$, $\eta=1$, $T=\SI{300}{\second}$, and $\Tb=T/I=\SI{10}{\second}$  to also allow  $\Rx$ to move  and to reduce the transmission window compared to the drug delivery system.
		
					%%%%%%%%%%%%%%%%%%%%%%%%%%%%%%%%%%%%%%%%%
					\begin{figure*}[!tbp]
						\centering
						\begin{minipage}[t]{0.49\textwidth}\hspace*{-5 mm}
							\centering
							\resizebox{1.05\linewidth}{!}{
								\includegraphics[scale=0.55]{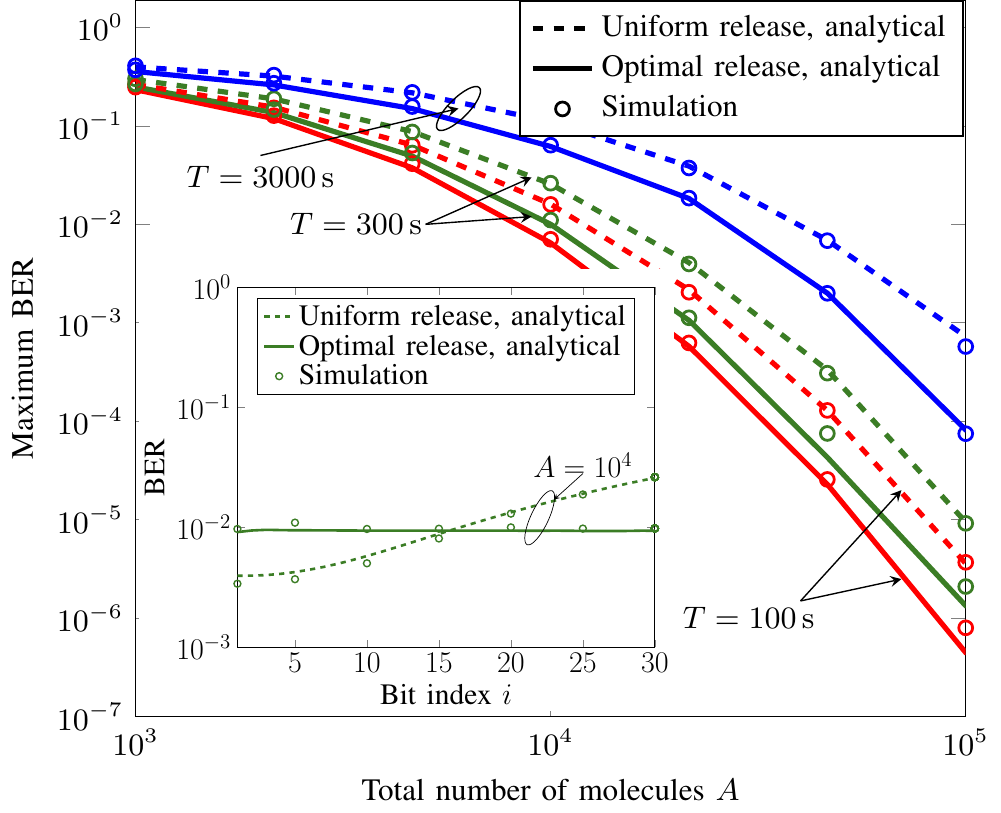}}\vspace*{-4 mm}
							\caption{
								Maximum  BER in a frame as a function of $A$ with uniform and optimal release. The inset shows the   BER for each bit in a frame for uniform  and optimal release for $A=10^4$ and $T=\SI{300}{\second}$. 	}	
							\label{fig:11}
						\end{minipage}
						\hfill
						\begin{minipage}[t]{0.1\textwidth}
						\end{minipage}
						\vspace*{-1 mm}
						\begin{minipage}[t]{0.49\textwidth}\hspace*{-5 mm}
							\centering
							\resizebox{1.05\linewidth}{!}{
								\includegraphics[scale=0.55]{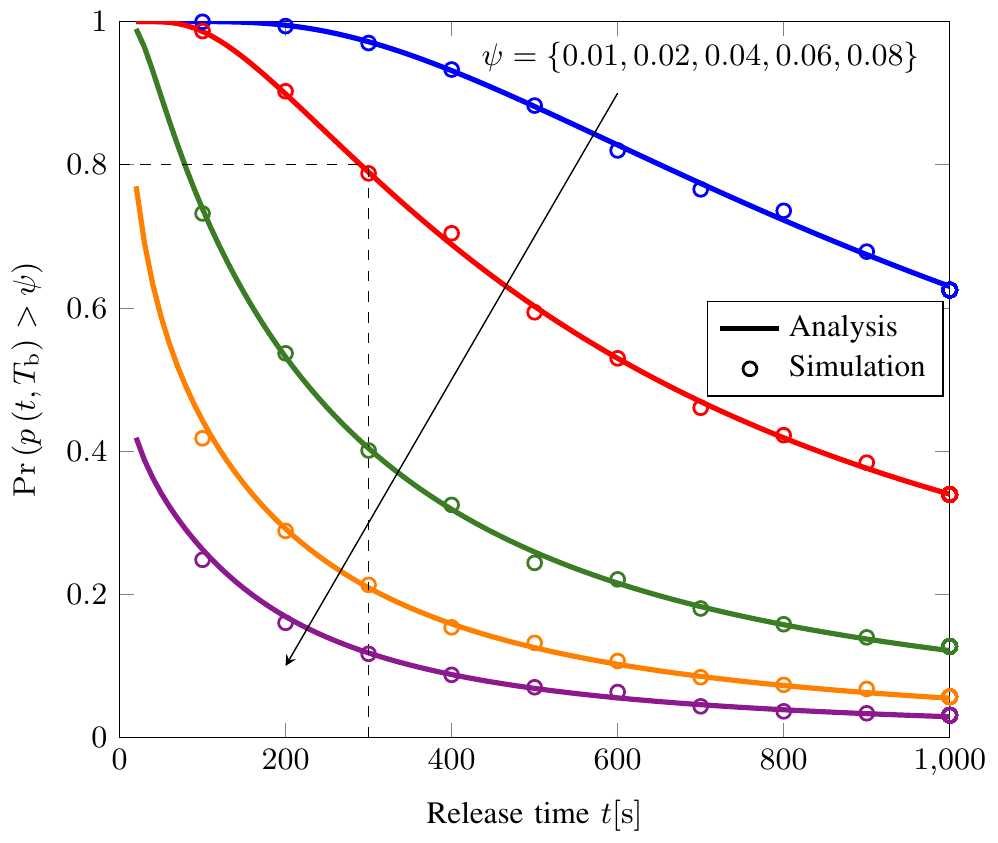}}\vspace*{-4 mm}
							\caption{
								The probability that  $p(t,\Tb)$  is larger than a given value $\psi$ as a function of $t$ for $\Tb=\SI{10}{\second}$. 	}	
							\label{fig:12}
						\end{minipage}
					\end{figure*}
					%%%%%%%%%%%%%%%%%%%%%%%%%%%%%%%%%%%%%%%%%%%%%%%%%%%%%%%%%%%	

		In Fig.~\ref{fig:10}, we consider the optimal release design, i.e., the optimal  number of  molecules available for transmission of  each bit in a frame, for an MC system with  fixed detection thresholds and fixed $A$, $A=\left\{10^3,10^4,10^5\right\}$.  The fixed detection thresholds $\xi$ are obtained from Subsection~\ref{sub5:2} by assuming uniform release.  Fig.~\ref{fig:10} reveals that in order to minimize the maximum  BER in a frame,  fewer molecules should be released at the beginning of the frame and  the number of released molecules  gradually increases with time. This is expected since, on average, for later release times,  more molecules are needed to compensate for the increasing distance  between the transceivers.

		Fig.~\ref{fig:11} shows the maximum  BER within a frame for uniform release and the proposed release design obtained from \eqref{eq:55a}, with a fixed detection threshold obtained from \eqref{eq:58},  as a function of $A$ for $T=\left\{100,300,3000\right\}\si{\second}$. As can be observed, the proposed optimal release profile leads to significant performance improvements   compared to uniform release, especially for large $A$.  For example, for $A=10^5$ and $T=\SI{3000}{\second}$, the maximum  BER is reduced by a factor of $8$ for optimal release compared to uniform release. On the other hand, to achieve a given desired BER, the total number of molecules $A$ required for optimal release is lower than that for uniform release.	In the inset of Fig.~\ref{fig:11}, we show the  BER as a function of bit index $i$ in one frame for uniform and optimal release for  $A=10^4$ and $T=\SI{300}{\second}$. We observe  that the optimal release achieves a lower maximum  BER compared to the uniform release.
		 We also observe that the optimal release leads to approximately the same  BER for each bit in a frame  which highlights the benefits of the proposed design. 
		
	Fig.~\ref{fig:12} shows the probability that $p(t, \Tb)$ is larger than a given value $\psi$, $\Pr\left(p(t, \Tb)>\psi\right)$, as a function of time $t$.  $\Pr\left(p(t, \Tb)>\psi\right)$ provides information about the probability that a released molecule is absorbed at the Rx, i.e.,  the efficiency of molecule usage. We observe from Fig.~\ref{fig:12} that $\Pr\left(p(t, \Tb)>\psi\right)$ is a decreasing function of time as expected from the analysis in Subsection~\ref{sub5:3}. Moreover, for a given $t$, $\Pr\left(p(t, \Tb)>\psi\right)$ is smaller for larger $\psi$. Furthermore, we can deduce the maximum time duration of a bit frame, $T^\star$, satisfying a required  molecule usage efficiency from Fig.~\ref{fig:12}. For example, for  $t\leq\SI{300}{\second}$,  $\Pr\left(p(t, \Tb)>0.02\right)\geq 0.8$ holds. Thus,  $T^\star=t+\Tb= \SI{310}{\second}$ guarantees a molecule usage efficiency of $\Pr\left(p(t, \Tb)>0.02\right)\geq 0.8$.  
	
	\section{Conclusions}\label{sec:7}
	In this paper, we considered a diffusive mobile MC system with an absorbing receiver, in which both the transceivers and the molecules diffuse. We provided a statistical analysis of the time-variant CIR and its integral, i.e., the probability that a  molecule is absorbed by $\Rx$ during  a given time period. We applied this statistical analysis to two system design problems, namely	drug delivery  and on-off keying based MC with imperfect CSI.  For the drug delivery system, we proposed an   optimal release profile which minimizes the number of released drug molecules while ensuring a target absorption rate for the drugs at the diseased site during a prescribed time period. The probability of satisfying the constraint on the absorption rate was adopted as a system performance criterion and  evaluated. We observed that ignoring the reality of the $\Tx$'s mobility for designing the release profile leads to unsatisfactory performance. For the MC system  with imperfect CSI, we optimized three design parameters, i.e., the detection threshold at  $\Rx$, the  release profile at  $\Tx$, and the time duration of a bit frame. Our simulation results revealed that the proposed MC system designs  achieved a better performance in terms of BER and molecule usage efficiency compared to a uniform-release system and a system without limitation on molecule usage, respectively. Overall, our results showed that taking into account the time-variance of the channel of mobile MC systems is crucial for achieving high performance.
	
%	\begin{appendices}
%		\section{Quantum Mechanics representations}
%		...
%		\section{Useful relations}
%		...
%	\end{appendices}
	\appendices
%	\section*{Appendices}
\renewcommand{\thesectiondis}[2]{\Alph{section}:}
	\section{Proof of Lemma~\ref{lem:0}}\label{app:0}
	To prove \eqref{eq:5c}, we have
	\begin{align} \label{eq:64}
	\mathrm{E}\left\{r(t)\right\}&\overset{(a)}{=}\sqrt{2 D_2 t} \mathrm{E}\left\{\gamma\right\}\overset{(b)}{=}\sqrt{2 D_2 t}\sqrt{2}\mathrm{e}^{-\frac{\lambda^2}{2}}\sum_{n=0}^{\infty}\frac{\Gamma\left(n+2\right)}{n!\Gamma\left(n+3/2\right)}\left(\frac{\lambda^2}{2}\right)^n\\ \nonumber
	&=\sqrt{\frac{D_2 t}{\pi}}4 \e^{-\frac{\lambda^2}{2}}\sum_{n=0}^{\infty}\frac{\left(n+1\right)!}{\left(2n+1\right)!}\left(2\lambda^2\right)^n=\sqrt{\frac{D_2 t}{\pi}}2 \e^{-\frac{\lambda^2}{2}}\sum_{n=0}^{\infty}\left(\frac{\left(n\right)!}{\left(2n\right)!}+\frac{\left(n\right)!}{\left(2n+1\right)!}\right)\left(2\lambda^2\right)^n\\\nonumber
	&\overset{(c)}{=}\sqrt{\frac{D_2 t}{\pi}}2 \e^{-\frac{\lambda^2}{2}}\left(1+\frac{\sqrt{\pi}}{2}\sqrt{2\lambda^2}\e^{\frac{\lambda^2}{2}}\erf\left(\frac{\lambda}{\sqrt{2}}\right)+\frac{\sqrt{\pi}}{\sqrt{2\lambda^2}}\e^{\frac{\lambda^2}{2}}\erf\left(\frac{\lambda}{\sqrt{2}}\right)\right),
	\end{align}
	where $\Gamma(\cdot)$ denotes the Gamma function, equality $(a)$ is  due to  \eqref{eq:4}, equality $(b)$ is due to \cite[Eq.~(1.5)]{Par:61:QAM}, and equality $(c)$ is obtained by applying  \cite[Eq.~(5.2.11.6) and Eq.~(5.2.11.7)]{PBM:86:Book}. Simplifying the final expression, we obtain \eqref{eq:5c}.
	
	To prove \eqref{eq:5d}, we have
	\begin{align} \label{eq:65a}
	\mathrm{E}\left\{r^2(t)\right\}&\overset{(a)}{=}2 D_2 t \mathrm{E}\left\{\gamma^2\right\}\overset{(b)}{=}2 D_2 t 2 \e^{-\frac{\lambda^2}{2}}\sum_{n=0}^{\infty}\frac{\Gamma\left(n+5/2\right)}{n!\Gamma\left(n+3/2\right)}\left(\frac{\lambda^2}{2}\right)^n\\ \nonumber	
		%&=4 D_2 t e^{-\frac{\lambda^2}{2}} \sum_{n=0}^{\infty}\frac{\left(n+3/2\right)}{n!}\left(\frac{\lambda^2}{2}\right)^n\\ \nonumber
		&=4 D_2 t e^{-\frac{\lambda^2}{2}}\left( \frac{\lambda^2}{2}\sum_{n=1}^{\infty}\frac{1}{(n-1)!}\left(\frac{\lambda^2}{2}\right)^{n-1}+\frac{3}{2}\sum_{n=0}^{\infty}\frac{1}{n!}\left(\frac{\lambda^2}{2}\right)^n\right)\\ \nonumber
		&\overset{(c)}{=}4 D_2 t e^{-\frac{\lambda^2}{2}}\left(\frac{\lambda^2}{2} e^{\frac{\lambda^2}{2}}+\frac{3}{2}e^{\frac{\lambda^2}{2}}\right) =r_0^2+6 D_2 t,
	\end{align}
	where equality $(a)$ is obtained due to  \eqref{eq:4}, equality $(b)$ is due to \cite[Eq.~(1.5)]{Par:61:QAM}, and equality $(c)$ is obtained by applying  the Maclaurin series of the  exponential function.
	From \eqref{eq:64} and \eqref{eq:65a}, we obtain \eqref{eq:5d} since $\mathrm{Var}\left\{r(t)\right\}=\mathrm{E}\left\{r^2(t)\right\}-\mathrm{E}^2\left\{r(t)\right\}$.
	
	To prove \eqref{eq:5}, we have
					\begin{align} \label{eq:65}
					f_{r(t)}(r)\overset{(a)}{=}\frac{1}{\sqrt{2 D_2 t}}f_{\gamma}\left(\gamma\right)
					\overset{(b)}{=}\frac{r}{r_0  \sqrt{\pi D_2 t}}\exp\left(-\frac{r^2+r_0^2}{4 D_2 t}\right)\sinh\left(\frac{r_0 r}{2 D_2 t}\right),
					\end{align}
					where $f_{\gamma}\left(\gamma\right)$ is the PDF of  $\gamma$. Equality $(a)$ in \eqref{eq:65} exploits the fact that $\gamma$ is  a function of $r(t)$ \cite[Eq.~(5-16)]{PP:02:Book}. Equality $(b)$ in \eqref{eq:65} is obtained from the expression for PDF $f_{\gamma}\left(\gamma\right)$ \cite[Eq.~(1.6)]{MBB:58:QAM} and the relation $I_{1/2}(x)=\sqrt{\frac{2}{\pi x}}\sinh(x)$, where $I_{1/2}(x)$ is the Bessel function of the first kind and order $1/2$.
					
					Moreover, since $\frac{r(t)}{\sqrt{2D_\Tx t}}$ follows a noncentral chi distribution, we obtain \eqref{eq:29} as \cite[Eq.~(1)]{Rob:69:Bel}
					\begin{align} \label{eq:31}
					F_{r(t)}(r)=F_{\frac{r(t)}{\sqrt{2D_2 t}}}\left(\frac{r}{\sqrt{2D_2 t}}\right)=1-\mathbf{Q}_{\frac{3}{2}}\left(\lambda,\frac{r}{\sqrt{2D_2 t}}\right).
					\end{align}
	\section{Proof of Theorem~\ref{theo:3}}\label{app:1}
	
	For this proof, we keep in mind that $\hat{h}\left(r,\tau\right)$ and $h\left(t,\tau\right)$ are two functions of different  variables but give the same value $h$ since $r$ is a function of $t$. Taking the derivative of \eqref{eq:6} with respect to $r$, we obtain \eqref{eq:27}. From \eqref{eq:27}, we observe that $\hat{h}'(r,\tau)=0$ is equivalent to  a cubic equation in $r$, given by  $a r^3+b r^2+ c r +d=0$, with properly defined coefficients $a$, $b$, $c$, $d$ and  discriminant $\Delta=18abcd-4b^3 d+b^2 c^2-4a c^3-27a^2 d^2$. From \eqref{eq:27}, we have $\Delta<0$ and thus $\hat{h}'(r,\tau)=0$ has only one real valued solution, denoted by $r^\star$, which corresponds to the maximum value of $\hat{h}(r,\tau)$, denoted by $h^\star$.  Then, from \eqref{eq:27}, we observe that $\hat{h}'(r,\tau)>0$ for $r<r^\star$ and $\hat{h}'(r,\tau)<0$ for $r>r^\star$. Therefore, the equation $\hat{h}\left(r,\tau\right)=h$ has two solutions $r_1(h)$ and $r_2(h)$, $r_1(h)< r_2(h)$, when $h<h^\star$, and has only one solution $r^\star$ when $h=h^\star=\hat{h}\left(r^\star,\tau\right)$. Finally, we derive \eqref{eq:25} by exploiting \cite[Eq.~(5-16)]{PP:02:Book} for the PDF of  functions of  random variables.  Moreover, for $ h=h^\star$, $\hat{h}'(r,\tau)=0$ so $f_{h\left(t,\tau\right)}(h)=\frac{f_{r(t)}(r^\star)}{\hat{h}'(r^\star,\tau)}\rightarrow \infty$.
	
		\section{Proof of Lemma~\ref{lem:2}}\label{app:3}
		To prove Lemma~\ref{lem:2}, we need to prove that $ \erf\left(\zeta_i(\xi,\alpha_i)\right)-\erf\left(\frac{\xi-\eta}{\sqrt{2\eta}}\right)$ is convex in $\xi$.  $\erf(\cdot)$ is a convex and non-decreasing function  for negative arguments and  a concave and non-decreasing function  for positive arguments. For $\eta<\xi<\mu_{i,1}$, we have $\zeta_i(\xi,\alpha_i)<0$ and $\frac{\xi-\eta}{\sqrt{2\eta}}>0$. Moreover, $\zeta_i(\xi,\alpha_i)$ and $\frac{\xi-\eta}{\sqrt{2\eta}}$ are affine functions of $\xi$. Then, $ \erf\left(\zeta_i(\xi,\alpha_i)\right)$ is  a convex and non-decreasing function of an affine function and thus is convex in $\xi$ \cite[Eq.~(3.10)]{BV:04:Book}. $\erf\left(\frac{\xi-\eta}{\sqrt{2\eta}}\right)$ is a concave and non-decreasing function of an affine function  and thus is concave in $\xi$ \cite[Eq.~(3.10)]{BV:04:Book}. Therefore, $ \erf\left(\zeta_i(\xi,\alpha_i)\right)-\erf\left(\frac{\xi-\eta}{\sqrt{2\eta}}\right)$ is convex, which concludes the proof.
		
	\section{Proof of Lemma~\ref{lem:1}}\label{app:2}
	To prove Lemma~\ref{lem:1}, we need to prove that $ \erf\left(\zeta_i(\xi,\alpha_i)\right)$ is convex in $\bm{\alpha}$. First, $\erf(\cdot)$ is a convex and non-decreasing function for negative arguments and $\zeta_i(\xi,\alpha_i)<0$ for $\xi<\mu_{i,1}$. Second, when $\eta<\xi$, $\zeta_i(\xi,\alpha_i)$ is a convex function in $\bm{\alpha}$ since its Hessian matrix  is positive semi-definite. Then, $ \erf\left(\zeta_i(\xi,\alpha_i)\right)$ is a convex and non-decreasing function of a convex function, and thus, it is convex in $\bm{\alpha}$ for $\eta<\xi<\mu_{i,1}$ \cite[Eq.~(3.10)]{BV:04:Book}, which concludes the proof.

	\section{Proof of Lemma~\ref{lem:3}} \label{app:4}
	To prove Lemma~\ref{lem:3}, we need to prove $\frac{\partial F_{r(t)}(r)}{\partial t}<0$. Using the Taylor series expansion of the $\sinh(\cdot)$  function in \eqref{eq:5}, we obtain
	\begin{align}
	\frac{\partial F_{r(t)}(r)}{\partial t}=&\frac{\partial}{\partial t}\left(\int_0^r\frac{x}{r_0  \sqrt{\pi D_2 t}}\exp\left(-\frac{x^2+r_0^2}{4 D_2 t}\right)\sum_{n=0}^\infty \frac{1}{(2n+1)!}\left(\frac{r_0 x}{2 D_2 t}\right)^{2n+1} \diff x\right)\\\nonumber
	=&\int_0^r\frac{x}{r_0  \sqrt{\pi D_2 }}
	\sum_{n=0}^\infty  \left[\frac{1}{(2n+1)!}\left(\frac{r_0 x}{2 D_2 }\right)^{2n+1}\frac{\partial}{\partial t}\left(\exp\left(-\frac{x^2+r_0^2}{4 D_2 t}\right)t^{-2n-3/2}\right)\right]\diff x\\\nonumber
	=&\int_0^r\frac{x}{r_0  \sqrt{\pi D_2 }}
	\sum_{n=0}^\infty  \left[\frac{1}{(2n+1)!}\left(\frac{r_0 x}{2 D_2 }\right)^{2n+1}\exp\left(-\frac{x^2+r_0^2}{4 D_2 t}\right)t^{-2n-5/2}\right.\\\nonumber
	&\left.\times\left(-2n-3/2+\frac{x^2+r_0^2}{4 D_2 t}\right)\right]\diff x \leq 0.
	\end{align}
	\bibliographystyle{IEEEtran}
	\bibliography{IEEEabrv,MolecularBib}

\end{document}